\documentclass[lettersize,journal]{IEEEtran}
\usepackage[top=1in, bottom=1.04in, left=1.01in, right=1.01in]{geometry}
\usepackage{setspace}
\usepackage{fancyhdr}

\ifCLASSINFOpdf
\else
\fi
\usepackage{amsmath}
\usepackage{amssymb}
\usepackage{mathtools}
\usepackage{enumitem}
\usepackage{amsthm}
\usepackage{bbm}
\usepackage{cite}
\usepackage{algorithm}
\usepackage{algorithmic}

\usepackage{xcolor}

\ifCLASSOPTIONcompsoc
 \usepackage[caption=false,font=normalsize,labelfont=sf,textfont=sf]{subfig}
\else
 \usepackage[caption=false,font=footnotesize]{subfig}
\fi
\allowdisplaybreaks
\newcommand{\bfa}{\mathbf{a}}
\newcommand{\bfb}{\mathbf{b}}

\newcommand{\bfg}{\mathbf{g}}
\newcommand{\bfh}{\mathbf{h}}

\newcommand{\bfs}{\mathbf{s}}

\newcommand{\bfw}{\mathbf{w}}
\newcommand{\bfx}{\mathbf{x}}
\newcommand{\bfy}{\mathbf{y}}
\newcommand{\bfz}{\mathbf{z}}

\newcommand{\calA}{\mathcal{A}}

\newcommand{\calB}{\mathcal{B}}

\newcommand{\calC}{\mathcal{C}}

\newcommand{\calD}{\mathcal{D}}

\newcommand{\calM}{\mathcal{M}}

\newcommand{\calN}{\mathcal{N}}

\newcommand{\calO}{\mathcal{O}}

\newcommand{\calH}{\mathcal{H}}

\newcommand{\bbC}{\mathbb{C}}
\newcommand{\bbR}{\mathbb{R}}
\newcommand{\bbN}{\mathbb{N}}
\newcommand{\bbE}{\mathbb{E}}
\newcommand{\bbZ}{\mathbb{Z}}
\newcommand{\nn}{\nonumber}

\newcommand{\bsth}{\boldsymbol{\theta}}

\newcounter{relctr} 

\newcommand\lr[2]{%
  \begingroup
    \refstepcounter{relctr}%
    ~\underset{\textnormal{(\alph{relctr})}}{\mathstrut{#1}}~%
    \originallabel{#2}%
  \endgroup
}
\AtBeginDocument{\let\originallabel\label} 

\newtheorem{theorem}{Theorem}
\newtheorem{lemma}{Lemma}

\fancypagestyle{firstpage}
{
    \fancyhf{}
    \fancyhead[C]{\footnotesize This work has been submitted to the IEEE for possible publication. Copyright may be transferred without notice, after which this version may no longer be accessible.}
    
}

\begin{document}

%
\title{Federated Learning Meets Fluid Antenna: Towards Robust and Scalable Edge Intelligence}
%
%
%

\author{Sangjun~Park,~\IEEEmembership{Member,~IEEE,}
        and Hyowoon~Seo,~\IEEEmembership{Member,~IEEE,}
		\thanks{
			S.~Park is with the Department of Electronics and Communications Engineering, Kwangwoon University, Seoul 01897, Korea (e-mail: sangjunpark@kw.ac.kr). 
		}
        \thanks{H.~Seo is with the Department of Electrical and Computer Engineering, Sungkyunkwan University, Suwon 16419, South Korea (e-mail: hyowoonseo@skku.edu)}
        \thanks{(\emph{Corresponding Author: Hyowoon Seo})}
        \vspace{-20pt}
}
\maketitle

\thispagestyle{firstpage}

\begin{abstract}
Federated learning (FL) is an emerging machine learning paradigm with immense potential to support advanced services and applications in future industries. However, when deployed over wireless communication systems, FL suffers from significant communication overhead, which can be alleviated by integrating over-the-air computation (AirComp). Despite its advantages, AirComp introduces learning inaccuracies due to the inherent randomness of wireless channels, which can degrade overall learning performance. To address this issue, this paper explores the integration of fluid antenna systems (FAS) into AirComp-based FL to enhance system robustness and efficiency. Fluid antennas offer dynamic spatial diversity by adaptively selecting antenna ports, thereby mitigating channel variations and improving signal aggregation. Specifically, we propose an antenna selection rule for fluid-antenna-equipped devices that optimally enhances learning robustness or training performance. Building on this, we develop a learning algorithm and provide a theoretical convergence analysis. The simulation results validate the effectiveness of fluid antennas in improving FL performance, demonstrating their potential as a key enabler for wireless AI applications.
\end{abstract}


%
\IEEEpeerreviewmaketitle

\section{Introduction}

The rapid advancement of machine learning (ML) has driven significant transformations across industries, with large language models (LLMs) like ChatGPT showcasing the vast potential of artificial intelligence. These models have revolutionized natural language processing, healthcare, and autonomous systems by providing real-time insights and enabling data-driven decision-making \cite{arXiv15_Konecny, SIGSAC15_Shokri, AISTATS17_Mcmahan}. As ML adoption continues to expand, ensuring efficiency, privacy, and scalability in model deployment has become increasingly critical.

One of the primary challenges in large-scale ML deployment is the substantial computational power and data transfer requirements. Traditional ML approaches rely on centralized training, where data from multiple sources is aggregated at a central server for model optimization. However, with the rise of distributed computing environments—including mobile devices, IoT networks, and edge servers—concerns related to data privacy, bandwidth limitations, and latency have become more pronounced \cite{WFL01, WFL02, WFL03, WFL04, WFL05}.

Federated learning (FL) presents a promising alternative by enabling decentralized model training, allowing devices to update models locally and share only the necessary updates instead of raw data \cite{arXiv15_Konecny, SIGSAC15_Shokri, AISTATS17_Mcmahan}. This approach mitigates privacy risks and alleviates communication overhead. However, deploying FL over wireless networks presents additional challenges, particularly in communication efficiency and scalability \cite{CommFL01, CommFL02, Mine01, CommFL03}.

In order to improve communication efficiency in FL, over-the-air computation (AirComp) has been introduced \cite{AFL01, AFL02, AFL03, AFL04, AFL05,Mine02,Mine03,Mine04}. Using the superposition property of wireless signals, AirComp enables the simultaneous transmission and aggregation of model updates, reducing latency and improving scalability \cite{AirCompFL01, AirCompFL02, AirCompFL03}. However, over-the-air FL—the integration of AirComp and FL—remains susceptible to signal interference and noise, limiting its practical effectiveness.

To address these challenges, this work explores the integration of fluid antenna systems (FAS) \cite{FAS} into over-the-air FL. Recognized as a novel antenna paradigm, FAS enables dynamic antenna repositioning, offering spatial diversity that enhances signal quality and learning performance \cite{FAS, FAS_p2p01, FAS_p2p02, FAS_p2p03, FAS_p2p04, FAS_MAC01, FAS_MAC02, FAS_MAC03, AirCompFAS01, AirCompFAS02}. By incorporating FAS into over-the-air FL, this paper proposes a novel framework, dubbed FAS-aided over-the-air FL (FAir-FL), which enhances both accuracy and robustness in distributed learning over wireless networks.

\subsection{Related Works}
Naturally, the key building blocks of the FAir-FL framework are wireless FL and FAS. While both technologies have been actively researched independently, their integration has yet to be thoroughly explored. However, FAir-FL is built upon a solid foundation of existing studies. The following summarizes the related work.

\textbf{Wireless Federated Learning} \quad
FL has been extensively studied in wireless networks due to its ability to facilitate collaborative model training while preserving data privacy \cite{WFL01, WFL02, WFL03, WFL04, WFL05}. Cellular networks naturally complement FL, as base stations can coordinate model updates among distributed devices. However, a major challenge remains: the high communication cost associated with transmitting local model updates (LMUs) \cite{CommFL01, CommFL02, Mine01, CommFL03}.

To address this issue, researchers have explored various optimization techniques such as model compression, quantization, and sparsification \cite{CommFL01, CommFL02, Mine01}. For instance, predictive coding-based compression and binary neural networks have been proposed to reduce transmission overhead \cite{CommFL02}. Additionally, high-altitude platforms and satellite-based FL systems have been investigated as potential solutions to enhance communication efficiency in large-scale deployments \cite{CommFL03}.

Among these advancements, AirComp has emerged as a key technology for alleviating communication bottlenecks in FL by enabling simultaneous model update transmissions \cite{AFL01, AFL02, AFL03, AFL04, AFL05, Mine02, Mine03, Mine04}. Recent studies have focused on optimizing over-the-air FL by addressing challenges such as interference management, adaptive aggregation mechanisms, and integration with digital twin (DT) edge networks \cite{AirCompFL01, AirCompFL02, AirCompFL03}.

Despite these developments, several open challenges remain. Most existing studies on FL aggregation fail to account for fluctuating wireless channel conditions, which can significantly impact model convergence and overall performance. Moreover, current over-the-air FL methods generally assume static antenna configurations, overlooking the potential benefits of dynamically adjustable antennas for mitigating interference and enhancing signal reception. These research gaps underscore the need for innovative solutions that improve the robustness and efficiency of FL in realistic wireless environments.

\textbf{Fluid Antenna Systems} \quad
FAS have recently gained interest as an alternative to conventional fixed-antenna designs, offering improved spatial diversity and adaptability \cite{FAS}. Unlike traditional multi-antenna systems, FAS enables dynamic antenna repositioning, making it particularly beneficial for space-constrained devices such as smartphones and IoT nodes \cite{FAS_p2p01, FAS_p2p02, FAS_p2p03, FAS_p2p04}.

The advantages of FAS extend to both single-user and multi-user scenarios, where dynamic antenna positioning enhances signal reception, reduces interference, and improves overall network performance \cite{FAS_MAC01, FAS_MAC02, FAS_MAC03}. Given its flexibility, recent research has explored its integration into over-the-air FL systems. Initial findings suggest that FAS can significantly boost communication efficiency and learning performance through optimized transceiver designs and intelligent antenna selection \cite{AirCompFAS01, AirCompFAS02}. However, further studies are needed to evaluate its scalability and real-world applicability.

Althoughugh FAS has shown promise in wireless communications, its application in over-the-air FL remains an underexplored area. Current research lacks a comprehensive analysis of how dynamic antenna positioning can enhance FL performance in noisy, interference-prone environments. Furthermore, the trade-off between antenna reconfiguration complexity and learning efficiency has not yet been fully investigated. These limitations serve as the main motivation for this work, which aims to bridge this research gap by developing FAir-FL, a framework designed to enhance wireless distributed learning through FAS-assisted over-the-air FL.

\subsection{Contribution and Organization}
The main contributions of this work are as follows:

\begin{itemize}
\item We develop a system model that integrates FAS with over-the-air FL, incorporating an adaptive antenna selection strategy and a detailed analysis of communication channels. Additionally, we propose a learning framework that accounts for both local training and global model aggregation.

\item We introduce an adaptive antenna selection mechanism tailored for FAir-FL. This mechanism includes two primary selection strategies—one focused on robustness and the other on model accuracy. Furthermore, we design a dynamic switching protocol that adjusts selection rules based on noise levels at the parameter server.

\item  We provide a rigorous mathematical analysis demonstrating the convergence properties of our FAir-FL framework.

\item Our results show superior performance compared to conventional fixed-antenna over-the-air FL in terms of communication efficiency, robustness, and model accuracy.

\end{itemize}

The rest of this paper is structured as follows. Section II provides the necessary preliminaries. Section III provides a detailed description of the proposed FAir-FL. Section IV offers numerical results validating the proposed method. Finally, Section VI concludes the paper and discusses potential directions for future research.

\emph{Notations:}
Throughout the article, scalars are written in a normal font, and vectors are written in a bold font. $\bbN$, $\bbZ^+$ $\bbR$ and $\bbC$ are sets of natural, positive integer, real and complex numbers, respectively. The $L^2$-norm is denoted as $\lVert\cdot\rVert$ for vectors and $\lvert\cdot\rvert$ for scalars, while $(\cdot)^{\mathsf{T}}$ represents the transpose. The element-wise product is represented by $\odot$. $\mathbf{0}_N$ and $\mathbf{1}_N$ denote $[0,\cdots,0]^{\mathsf{T}}\in\bbR^N$ and $[1,\cdots,1]^{\mathsf{T}}\in\bbR^N$, respectively, and $\mathbf{I}_N$ denotes the $N \times N$ identity matrix. For $a\in\bbN$, $[a]$  denotes $\{n \mid n\in\bbN, 1\le n\le a\}$. For set $\calA\subset\bbC$ and vector $\bfa=[a_1,\cdots,a_N]^{\mathsf{T}}\in\bbC^N$, $\mathbbm{1}_\calA(\bfa)\in\{0,1\}^N$ denotes element-wise indicator function where its $i$-th element is $1$ if $a_i\in\calA$ and $0$ otherwise.

\section{Preliminaries}
To present the core concepts of the proposed Fair-FL, this section provides an overview of the key theoretical background and the system model considered in this study.

\subsection{Theoretical Background}
\subsubsection{Federated Learning (FL)}
FL is a distributed machine learning approach that allows multiple devices to collaboratively train a global model while keeping the data localized, enhancing data privacy and security. The training process per epoch is organized as follows:

\begin{enumerate}
\item \textit{Local Training} - Devices train locally using their data, employing algorithms like stochastic gradient descent (SGD) to update model parameters. This retains data privacy by keeping sensitive information on-device.

\item \textit{Local Model Update Transmission} - Post-training, devices send their LMUs, encapsulating learned insights, to a central server.

\item \textit{Global Model Aggregation} - The central server uses methods like federated averaging (FedAvg) to merge these LMUs into a global model update (GMU), thus integrating knowledge from all nodes.

\item \textit{Global Model Update Dissemination} - The global model is then distributed to all devices, marking the beginning of the next training cycle. This iterative process aims to enhance model performance over time.
\end{enumerate}

\subsubsection{Fluid Antenna Systems (FAS)}
Fluid Antenna Systems (FAS) dynamically adjust their antenna positions within a device to optimize signal reception. Using either conductive liquids or arrays of antenna ports, FAS adapt to real-time channel conditions by selecting the best reception point, thereby minimizing interference and maximizing signal strength.

The compact arrangement of antenna ports within the device ensures channel correlation. The channel model for a single fluid antenna with $U$ ports is expressed as:
\begin{align}
h_{u}\! =\! \frac{1}{\sqrt{2}}\sqrt{1-\mu_u^2}(a_u+jb_u) + \frac{1}{\sqrt{2}}\mu_u (a_1+j b_1),\label{eq:FAS_channel}
\end{align}
where $\{a_u\}$ and $\{b_u\}$ are independent normal variables, and $\mu_u$ quantifies the port correlation, calculated using:
\begin{align}
\mu_u = J_0\bigg( \frac{2\pi (u-1)}{U-1}W \bigg),
\end{align}
with $J_0$ as the zero-order Bessel function, and $W$ as the normalized distance between ports. This model highlights how FAS can adaptively mitigate fading and interference, particularly beneficial in dynamic multipath or interference-heavy environments, thus improving spectral efficiency and communication reliability.

\subsection{System Model}  
\begin{figure}[t]
    \centering
    \subfloat{\includegraphics[width=0.48\textwidth]{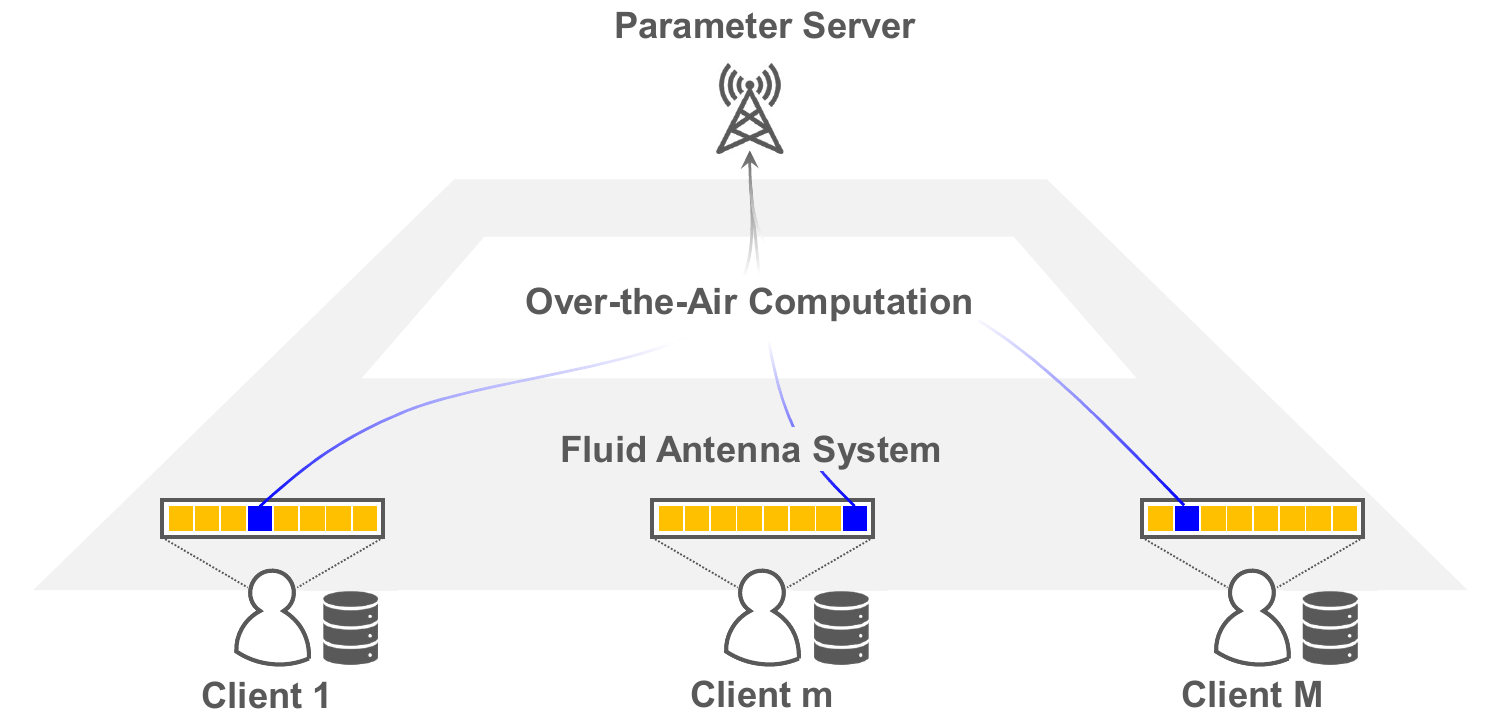}} 
    \caption{An illustration of FAS-aided over-the-air (FAir)-FL system \label{fig:systemmodel}}
\end{figure}
Throughout this paper, we consider a centralized system consisting of a single parameter server (hereafter referred to as a server) and $M \in \mathbb{Z}^{+}$ distributed wireless devices, each equipped with a single fluid antenna, as illustrated in Fig. \ref{fig:systemmodel}. Each device $m \in [M]$ maintains a local dataset $\mathcal{D}_m = \{\mathbf{s}_m^i \mid i \in [D_m] \}$ containing $D_m \in \mathbb{Z}^{+}$ training samples, where $\mathbf{s}_m^i$ represents the $i$-th training sample of device $m$.

The local loss function for device $m$ is defined as
\begin{align}
L_m(\boldsymbol{\theta}) = \frac{1}{D_m} \sum_{i=1}^{D_m} \ell(\mathbf{s}_m^i; \boldsymbol{\theta}),
\end{align}
where $\ell( \cdot ; \boldsymbol{\theta})$ is a sample-wise loss function parameterized by $\boldsymbol{\theta} \in \mathbb{R}^d$. In our experiments, we adopt cross-entropy as the sample-wise loss function, as the primary focus is on classification tasks.

The system's objective is to determine the optimal parameter $\boldsymbol{\theta}^* \in \mathbb{R}^{d}$ that minimizes the overall loss function, given by
\begin{align}
L(\boldsymbol{\theta}) = \sum_{m = 1}^{M} \rho_m L_m(\boldsymbol{\theta}),
\end{align}
where $\rho_m = D_m / \sum_{m = 1}^{M} D_m$ represents the proportion of the total training samples contributed by device $m$.

Similar to other FL algorithms, the system operates in global and local rounds. Global rounds are defined from the perspective of the GMUs, while local rounds correspond to LMUs. For simplicity, we assume that FAir-FL runs for $T \in \mathbb{Z}^{+}$ global rounds and $E \in \mathbb{Z}^{+}$ local rounds.

In each global round $t \in [T]$, communication takes place over $N \in \mathbb{Z}^{+}$ OFDM symbols, each comprising $F \in \mathbb{Z}^{+}$ subcarriers. Under the FAS channel model described in \eqref{eq:FAS_channel}, the channel from port $u$ of device $m$ at the $f$-th subcarrier in the $n$-th OFDM symbol is expressed as
\begin{align} \nonumber
h_{m,u,f}^{t,n} = \frac{1}{\sqrt{2}}& \sqrt{1 - \mu_u^2}\big(a_{m,u,f}^{t,n} + j b_{m,u,f}^{t,n} \big)\\
&\ \ \quad \quad + \frac{1}{\sqrt{2}} \mu_u \big(a_{m,1,f}^{t,n} + j b_{m,1,f}^{t,n} \big),
\end{align}
where $a_{m,u,f}^{t,n}$ and $b_{m,u,f}^{t,n}$, $\forall u\in[U]$ are independently and normally distributed random variables. Furthermore, we define
\begin{align}
\mathbf{h}_{m,u}^{t,n} = [h_{m,u,1}^{t,n},h_{m,u,2}^{t,n}, \dots,h_{m,u,F}^{t,n}]^{\mathsf{T}} \in \mathbb{C}^{F}
\end{align}
as the channel vector from port $u$ of device $m$ in the $n$-th OFDM symbol during the $t$-th global round.

The proposed FAir-FL integrates the characteristics of FAS into the conventional FL framework by incorporating an antenna port selection mechanism in each global round. Let $u_m^{t,n}$ denote the optimal antenna port selected for the $n$-th OFDM symbol on device $m$ during the $t$-th global round. For simplicity, we define the corresponding channel vector as
\begin{align}
\mathbf{h}_{m,u_{m}^{t,n}}^{t,n} \equiv \tilde{\mathbf{h}}_{m}^{t,n} = \left[\tilde{h}_{m,1}^{t,n},\tilde{h}_{m,2}^{t,n}, \dots,\tilde{h}_{m,F}^{t,n}\right]^{\mathsf{T}},
\end{align}
where that $\tilde{h}_{m,f}^{t,n} \equiv h_{m,u_m^{t,n},f}^{t,n}$, $\forall f \in [F]$.

Each OFDM symbol, denoted by $\bfx_{m}^{t,n} \in \mathbb{C}^{F}$ for some $m \in [M]$, $t \in [T]$, and $n \in [N]$, can convey $F$ real-valued numbers. We assume that transmission power of each OFDM symbol is bounded by $P$, that is  \begin{align}
    \lVert\bfx_m^{t,n}\rVert^2 \leq P,\quad \forall (m,t,n) \in [M]\times[T]\times[N].
\end{align}
Consequently, the received signal at the server, which is a superposition of signals from all $M$ devices, is given by
\begin{align}
\mathbf{y}^{t,n} = \sum_{m=1}^{M} \tilde{\mathbf{h}}_{m}^{t,n} \odot \mathbf{x}_{m}^{t,n} + \mathbf{w}^{t,n},
\end{align}
where $\mathbf{w}^{t,n} \sim \mathcal{CN}(\mathbf{0}_F, N_0 \mathbf{I}_F)$ represents additive noise with noise power $N_0$.

\begin{figure*}[t]
    \centering
\subfloat{\includegraphics[width=\textwidth]{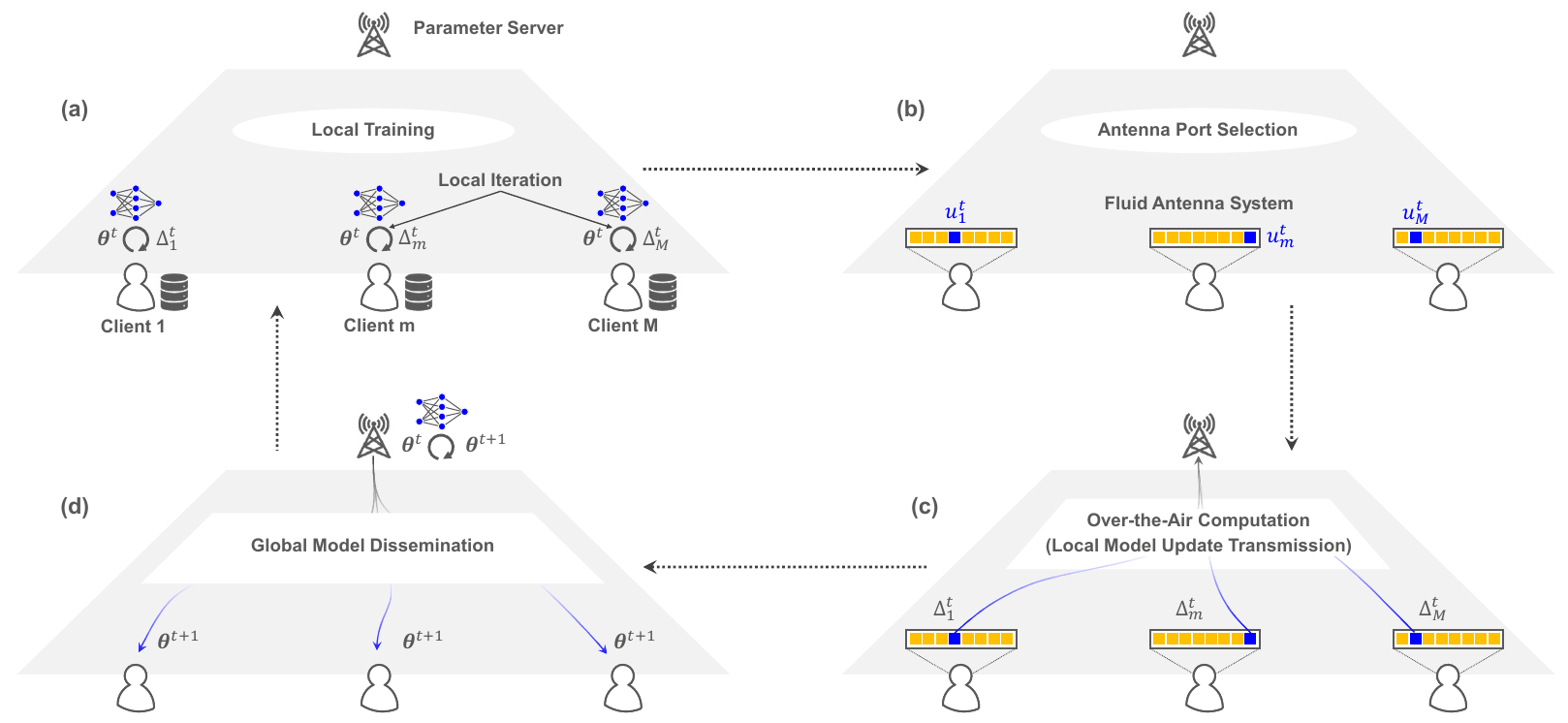}} 
    \caption{An illustration of the FAir-FL system's operational structure, depicted in sequential stages: (a) local training, (b) antenna port selection, (c) local model update aggregation using AirComp, (d) local model averaging and global model dissemination. This process is repeated until the model converges. \label{fig:operation}}
\end{figure*}

\section{Fluid Antenna System-aided Over-the-Air Federated Learning (FAir-FL)}
We now provide a detailed explanation of FAir-FL. First, we present a comprehensive description of the proposed algorithm, followed by a mathematical discussion on its convergence.

\subsection{Operational Structure of FAir-FL}
In essence, as shown in Fig. \ref{fig:operation}, FAir-FL operates as follows: At each global round, each device computes its LMU based on its local dataset. Considering the current channel conditions, it then selects the optimal antenna port and transmits the LMU to the server via the selected port using AirComp. The server aggregates the received LMUs to derive the GMU, which is subsequently used to update the global model for the next round.

Specifically elaborating on the update process, at the $t$-th global round, for some $t \in [T]$, device $m \in [M]$ updates the global parameter $\bsth^{t} \in \mathbb{R}^d$ and obtains its local parameter $\bsth_m^{t,E} \in \mathbb{R}^d$ after the $E$ local training rounds based on the batch-SGD algorithm. The parameter after $e$-th local round is represented as
    \begin{align}
        \bsth_m^{t,e}=\bsth_m^{t,e-1}-\lambda\bfg_m^{t,e}\in\bbR^d, \label{eq:localTrain}
    \end{align}
    for some $e\in[E]$, where 
    \begin{align}
        \bfg_m^{t,e}=\frac{1}{|\calB_m^{t,e}|}\sum_{\bfs\in\calB_m^{t,e}} \nabla l(\bfs;\bsth_m^{t,e-1})
    \end{align}
    is the gradient obtained in the $e$-th local round, and $\bsth_m^{t,0}=\bsth^{t}$, which means that the local training at the $t$-th global round is initialized by the global update $\bsth^t$. $\calB_{m}^{t,e}\subset \calD_m$ is the mini-batch at the $e$-th local round, whose size is fixed as $\lvert\calB_{m}^{t,e}\rvert=B$, $\forall m \in [M]$, $\forall t \in [T]$ and $\forall e \in [E]$.
    
    After the local updating process, the device $m$ transmits its LMU, defined by 
    \begin{align}
    \Delta_m^t =\bsth_m^{t,E} - \bsth^t =[\delta_{m,1}^t,\delta_{m,2}^t\cdots,\delta_{m,d}^t]^{\mathsf{T}},
    \end{align}
    to the parameter server through wireless medium.

    Suppose that the optimal antenna port $u_{m}^{t,n}$ is selected $\forall m \in [M]$.\footnote{The method for selecting the optimal antenna port will be discussed later in this section.} Then, the $n$-th OFDM symbol $\bfx_{m}^n$ for device $m$ is given by
    \begin{align} \bfx_{m}^{t,n} &= \gamma^{t,n} \rho_m \alpha_{m}^{t,n} \odot \mathbbm{1}_{\calH(\tau)}(\tilde{\bfh}_{m}^{t,n}) \odot \Delta_{m}^{t,n}, \label{eq:TXsymbol} \end{align}
    where $\gamma^{t,n}$ is a positive coefficient for power scaling, bounded as
    \begin{align} \gamma^{t,n} \le \frac{\sqrt{P}}{\beta^t} \min_{m} \left[ \sum_{f=1}^{F} {\underbrace{\lvert \tilde{h}_{m,f}^{t,n} \rvert^{-2} \mathbbm{1}_{\calH(\tau)}( \tilde{h}_{m,f}^{t,n})}_{\triangleq S_{m,f}^{t,n}}} \right]^{-\frac{1}{2}}, \label{ineq:gamma} \end{align}
    where $\beta^t$ denotes the maximum amplitude across all devices, defined as
    \begin{gather} \beta^t = \max_{m} \max_{j\in[d]} \lvert \delta_{m,j}^t \rvert. \label{def:beta} \end{gather}
    Here, $\calH(\tau)$ is the subset of complex numbers with magnitudes greater than or equal to the heuristically determined channel threshold $\tau$, defined as
    \begin{gather} \calH(\tau) = \{z\in\bbC \mid \lvert z\rvert \ge \tau\}. \end{gather}
    This means that only channels with gains exceeding $\tau$ are utilized for transmitting the LMU.
    The term $\alpha_m^{t,n} \in \bbC^F$ represents the element-wise reciprocal of $\tilde{\bfh}_m^{t,n}$, defined as
    \begin{gather} \alpha_{m}^{t,n} \odot \tilde{\bfh}_{m}^{t,n} = \mathbf{1}_F. \end{gather}
    Thus, $\alpha_m^{t,n}$ is used to compensate for the channel effects. The vector $\Delta_m^{t,n}$ represents the LMU of device $m$ transmitted through the $n$-th OFDM symbol, defined as
    \begin{gather} \Delta_{m}^{t,n} = [\delta_{m,(n-1)F+1}^t, \dots, \delta_{m,nF}^t]^{\mathsf{T}} \in \bbR^F. \end{gather}
    
    Notably, $\beta^t$ can be obtained through a simple one-round trip communication process: each device transmits its maximum LMU element magnitude, $\beta_{m}^t = \max_{j} \lvert \delta_{m,j}^t \rvert$, to the server. The server then computes the maximum value and broadcasts $\beta^t$ back to the devices.

    The transmit power of each device is bounded as
    \begin{align} \lVert\bfx_m^{t,n}\rVert^2 &\le (\gamma^{t,n})^2 \sum_{f=1}^{F} (\rho_m \delta_{m,(n-1)F+f}^t)^2 S_{m,f}^{t,n}, \end{align}
    which ensures compliance with the maximum transmit power constraint. The received signal can be expressed as
    \begin{align} \bfy^{t,n} = \gamma^{t,n} \sum_{m\in\calM} \rho_m \Delta_{m}^{t,n} \odot \mathbbm{1}_{\calH(\tau)}(\tilde{\bfh}_{m}^{t,n}) + \bfw^{t,n}. \label{eq:RXsymbol} \end{align}
    From the received signal, the aggregated update can be estimated as
    \begin{align} \hat{\Delta}^{t,n} &= \frac{\bfy^{t,n}}{\gamma^{t,n}}\\
    &= \sum_{m\in\calM} \rho_m \Delta_{m}^{t,n} \odot \mathbbm{1}_{\calH(\tau)}(\tilde{\bfh}_{m}^{t,n}) + \tilde{\bfw}^{t,n}, \label{eq:hatDel} \end{align}
    where $\tilde{\bfw}^{t,n} \sim \calC\calN(\mathbf{0}_F, N_0/(\gamma^{t,n})^2\mathbf{I}_F)$.

    The GMU is then constructed by concatenating the recovered aggregations:
    \begin{gather} \hat{\Delta}^t = [(\hat{\Delta}^{t,1})^{\mathsf{T}}, \dots, (\hat{\Delta}^{t,N})^{\mathsf{T}}]^{\mathsf{T}}. \end{gather}
    The corresponding additive noise is given by
    \begin{gather} \tilde{\bfw}^{t} = [\tilde{\bfw}^{t,1}, \dots, \tilde{\bfw}^{t,N}]^{\mathsf{T}}. \end{gather}

    Using the GMU, the server updates the global model as
    \begin{align} \bsth^{t+1} = \bsth^{t} + \hat{\Delta}^{t}, \end{align}
    and broadcasts the updated model to the devices. This iterative process continues to train the global model. The complete procedure is summarized in \textbf{Algorithm} \ref{alg:FAFL}.

\begin{algorithm}[t]
    \caption{FAir-FL Algorithm \label{alg:FAFL}}
    \begin{algorithmic}[1]
        \REQUIRE{$\bsth^1$}
        \ENSURE $\bsth^{T+1}$
        \FOR{$t\gets 1$ \TO $T$}
            \FOR{$m\gets1$ \TO $M$ in parallel}
                \STATE Device $m$ trains $\bsth_m^{t,E}$ from $\bsth^{t}$ as in \eqref{eq:localTrain} 
                \STATE Device $m$ computes $\Delta_m^t\gets \bsth^{t,E} - \bsth^{t}$
                \STATE Device $m$ finds $\beta_m^t\gets \max_j\lvert\delta_{m,j}^t \rvert$
                \STATE Device $m$ transmits $\beta_m^t$ to parameter server
            \ENDFOR
            \STATE Server obtains $\beta^t\gets\max_m \beta_m^t$
            \STATE Server broadcasts $\beta^t$ to devices
            \FOR{$n\gets1$ \TO $N$}
                \FOR{$m\gets1$ \TO $M$ in parallel}
                    \STATE Device $m$ selects antenna port $u_m^{t,n}$ as in \eqref{eq:selection_fin}
                    \STATE Device $m$ transmits $\bfx_m^{t,n}$ as in \eqref{eq:TXsymbol} to parameter server
                \ENDFOR
                \STATE Server receives $\bfy^{t,n}$ as in \eqref{eq:RXsymbol}
                \STATE Server recovers $\hat{\Delta}^{t,n}$ as in \eqref{eq:hatDel}
            \ENDFOR
            \STATE Server constructs $\hat{\Delta}^t\gets[(\Delta^{t,1})^{\mathsf{T}},\cdots,(\Delta^{t,N})^{\mathsf{T}}]^{\mathsf{T}}$
            \STATE Server updates $\bsth^{t+1}\gets\bsth^t+\hat{\Delta}^t$
            \STATE Server broadcasts $\bsth^{t+1}$ to devices
        \ENDFOR
    \end{algorithmic}
\end{algorithm}

\subsection{Convergence Analysis}

First, we adopt the following assumptions \textbf{(A1)}-\textbf{(A4)}, which are commonly used in convergence analysis within the literature \cite{Mine01, Mine02, Mine03, Mine04, Assumption}.

\noindent \textbf{(A1)} $L_m$ is a $\kappa$-smooth function, $\forall m\in [M]$, that is, $\forall (\bfx,\bfy) \in\bbR^{d\times d}$,
        \begin{align}
            \lVert\nabla L_m(\bfx)-\nabla L_m(\bfy) \rVert\le \kappa \lVert \bfx-\bfy \rVert, \label{eq:smooth}
        \end{align}
        or, equivalently,
        \begin{align}
            \lvert L(\bfy)-L(\bfx)-\nabla L(\bfx)^\mathsf{T}(\bfy-\bfx)\rvert
            \le \frac{\kappa}{2}\rVert\bfy-\bfx\lVert^2. \label{eq:smooth2}
        \end{align}

\noindent \textbf{(A2)} There exists $\sigma$ such that
        \begin{align}
            \bbE[{\lVert\nabla l(\bfs;\bsth)-\nabla L_m(\bsth)\rVert^2}] \le \sigma^2, \ \  \forall m \in [M].
            \label{eq:boundness1}
        \end{align}

\noindent \textbf{(A3)} $\bbE[\nabla l(\bfs;\bsth)] = \nabla L_m(\bsth)$,  $\ \forall m\in [M]$.
    
\noindent \textbf{(A4)} There exists $\delta^2$ such that
        \begin{align}
            \frac{1}{M}\sum_{m=1}^{M}\lVert \nabla L_m(\bsth) - \nabla L(\bsth) \rVert^2 \le \delta^2. \label{eq:diversity}
        \end{align}
Based on these assumptions, we can derive the following Theorem.
\begin{theorem} \label{lemma01}
    If $\lambda<\frac{e^{-\tau}}{32\kappa M E}$, then the following inequality holds.
    \begin{align}
    &\bbE\bigg[\frac{1}{T}\sum_{t=1}^{T}\lVert \nabla L(\bsth^{t}) \rVert^2 \bigg] 
    \le \frac{8(L(\bsth^1) - L(\bsth^*))}{\lambda e^{-\tau} E T} \nn \\
    &\qquad+(c_1+c_2)c_3+\frac{4\kappa N_0  F}{\lambda e^{-\tau} E T}\sum_{t=1}^{T}\sum_{n=1}^{N}\bbE[ ({\gamma^{t,n}})^{-2}], \label{eq:thm1}
\end{align}
where $c_1=32\lambda\kappa M E\sigma^2$, $c_2=32\lambda\kappa M^2 E \delta^2$ and $c_3=\lambda\kappa E + e^{\tau}$.
\end{theorem}
\begin{IEEEproof}
    The proof is provided in Appendix B.
\end{IEEEproof}
From the above result, we observe that if the right-hand side of the inequality \eqref{eq:thm1} becomes smaller, the proposed training algorithm results in faster convergence. Therefore, to achieve fast convergence, each device should properly select its antenna port to minimize the term $\bbE[(\gamma^{t,n})^{-2}]$, $\forall (t,n) \in [T]\times[N]$.

\subsection{Optimal Antenna Port Selection}
For the sake of simplicity and without loss of generality, we omit the superscripts $t$ and $n$, respectively representing global round and OFDM symbol indices, for the remainder of this section.

Recall that for given $u_m$, $\forall m\in [M]$, $\gamma$ is bounded as in \eqref{ineq:gamma}, with its maximum value simply denoted by $\bar{\gamma}$:
\begin{align} \gamma \leq \bar{\gamma} = \frac{\sqrt{P}}{\beta} \min_{m}\left[\sum_{f=1}^{F}S_{m,f} \right]^{-\frac{1}{2}}. \label{eq:maxgamma} \end{align}

As we aim to minimize $\bbE[\gamma^{-2}]$ for fast convergence, it is desirable to maximize $\bar{\gamma}$. Accordingly it is also desirable to minimize $\sum_{f=1}^{F} S_{m,f}$, $\forall m \in [M]$. Consequently, a promising approach for antenna port selection is:
\begin{align} u_{m}=\arg\min_{u} \sum_{f=1}^{F}\lvert h_{m,u,f} \rvert^{-2} \mathbbm{1}_{\calH(\tau)}( h_{m,u,f}). \label{eq:select01} \end{align}
We define $\bar{\gamma}_1$ analogously to \eqref{eq:maxgamma}, corresponding to the selection criterion in \eqref{eq:select01}.

On the other hand, considering that transmitting more LMUs enhances training performance, each device $m$ should transmit its parameters whenever the corresponding channel gains exceed a given threshold. In this case, the antenna port $u_m$ is selected as follows:
\begin{align} u_{m} &= \arg\max_u\sum_{f=1}^{F}\mathbbm{1}_{\calH(\tau)} (h_{m,u,f})\\ &=\arg\max_u \mathbf{1}_F^{\mathsf{T}}\mathbbm{1}_{\calH(\tau)} (\bfh_{m,u}) \label{eq:selection02} \end{align}
Similarly, we define $\bar{\gamma}_2$ as in \eqref{eq:maxgamma} for the selection strategy in \eqref{eq:selection02}.

The first selection rule \eqref{eq:select01} is designed to enhance robustness by reducing the additive noise power in the received GMU. By reducing the noise power, this rule ensures that the server receives a more stable LMU. 
In contrast, the second rule \eqref{eq:selection02} aims to improve accuracy by ensuring precise model updates. Since this rule allows a larger number of updates to be transmitted to the server, the global model can incorporate more diverse information which leads to improved learning accuracy.

To balance both robustness and accuracy, a hybrid approach is adopted, enabling an optimal selection strategy that adjusts dynamically based on the noise conditions. The choice between these two selection rules depends on the effective noise power at the parameter server, denoted as $N_0 / (2\bar{\gamma}_2)$. Specifically, the selection rule is determined as follows:
\begin{align} u_m&= \begin{cases} \text{r.h.s. of \eqref{eq:select01}},
&\frac{N_0}{2\bar{\gamma}_2} > \psi,\\ \text{r.h.s. of \eqref{eq:selection02}},
&\frac{N_0}{2\bar{\gamma}_2} \le \psi. \end{cases} \label{eq:selection_fin} \end{align}

Based on this selection criterion, we establish the following lemma.

\begin{lemma} \label{lemma02} Given $\bar{\gamma}$ in \eqref{eq:maxgamma} and the antenna port selection in \eqref{eq:selection_fin}, if $\lambda \leq \frac{1}{2\sqrt{2}\kappa E}$, then the expectation of $\bar{\gamma}^{-2}$ is bounded as follows:
\begin{align}
\bbE[\bar{\gamma}^{-2}] \le  \frac{8\lambda^2 M E^2 F \omega}{P\tau^2}
(   \sigma^2 +  M  \delta^2 + \bbE[\lVert \nabla L(\bsth^t) \rVert^2]).
\end{align}
where $\omega = 2 - (1-\hat{p})^{FMU}$ with $\hat{p} = e^{-\tau^2/2(1-\mu_{\min}^2)}$ and $\mu_{\min}=\min_u \mu_u$.
\end{lemma}
\begin{IEEEproof} The proof is provided in Appendix C. \end{IEEEproof}

Using Lemma \ref{lemma02}, we establish the convergence of the proposed algorithm in the following theorem.

\begin{theorem} \label{thm:thm01} Given $\bar{\gamma}$ in \eqref{eq:maxgamma} and the antenna port selection strategy in \eqref{eq:selection_fin}, if $\lambda<\min{\frac{P\tau^2 e^{-\tau}}{64\kappa N_0 M E F^2 N \omega},\frac{e^{-\tau}}{32\kappa M E}} $, then:
\begin{align}
    &\bbE\bigg[\frac{1}{T}\sum_{t=1}^{T}\lVert \nabla L(\bsth^{t}) \rVert^2 \bigg]  \nn \\
&\qquad\qquad\le \frac{16(L(\bsth^1) - L(\bsth^*))}{\lambda e^{-\tau} E T} +(\tilde{c}_1+\tilde{c}_2)\tilde{c}_3, \nn
\end{align}

where $\tilde{c}_1=64\lambda \kappa M E\sigma^2$, $\tilde{c}_2=64\lambda \kappa M^2 E \delta^2$, and $\tilde{c}_3=\lambda\kappa E + e^\tau + N_0 F^2 N \omega e^{\tau} P^{-1} \tau^{-2}$. 

Furthermore, for $\lambda=T^{-\frac{1}{2}}$,
\begin{align}
    &\bbE\bigg[\frac{1}{T}\sum_{t=1}^{T}\lVert \nabla L(\bsth^t)\bigg\rVert^2\bigg]=\calO(T^{-\frac{1}{2}}). \nn
\end{align}
\end{theorem}

\begin{IEEEproof} The proof is provided in Appendix D. \end{IEEEproof}

\section{Numerical results}
\begin{figure*}[t]
    \centering
    \subfloat[$P=-5$ {[}dBm{]}]{\includegraphics[width=0.24\textwidth]{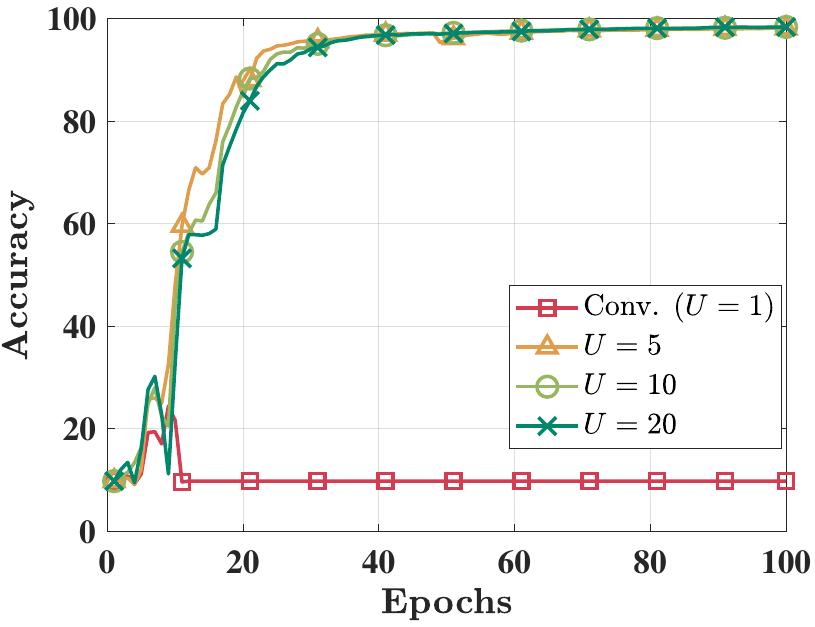}\label{fig:MNIST_P_m05}} 
    \hspace{1pt}
    \subfloat[$P=0$ {[}dBm{]}]{\includegraphics[width=0.24\textwidth]{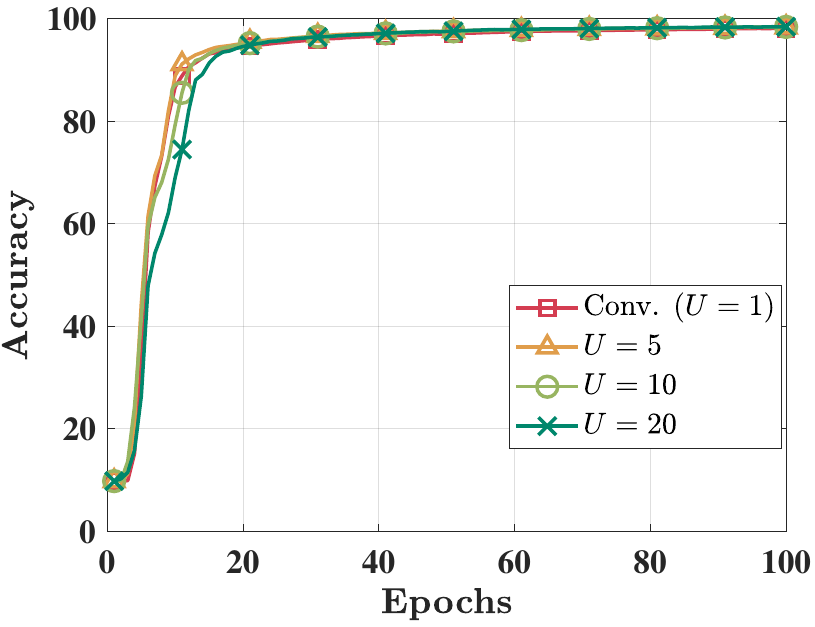}\label{fig:MNIST_P_000}} 
    \hspace{1pt}
    \subfloat[$P=5$ {[}dBm{]}]{\includegraphics[width=0.24\textwidth]{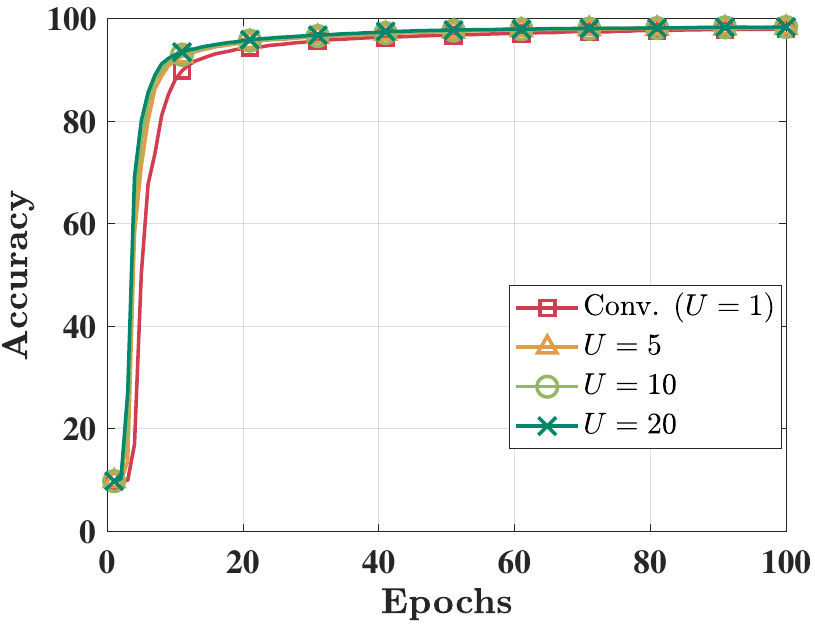}\label{fig:MNIST_P_005}} 
    \hspace{1pt}
    \subfloat[$P=10$ {[}dBm{]}]{\includegraphics[width=0.24\textwidth]{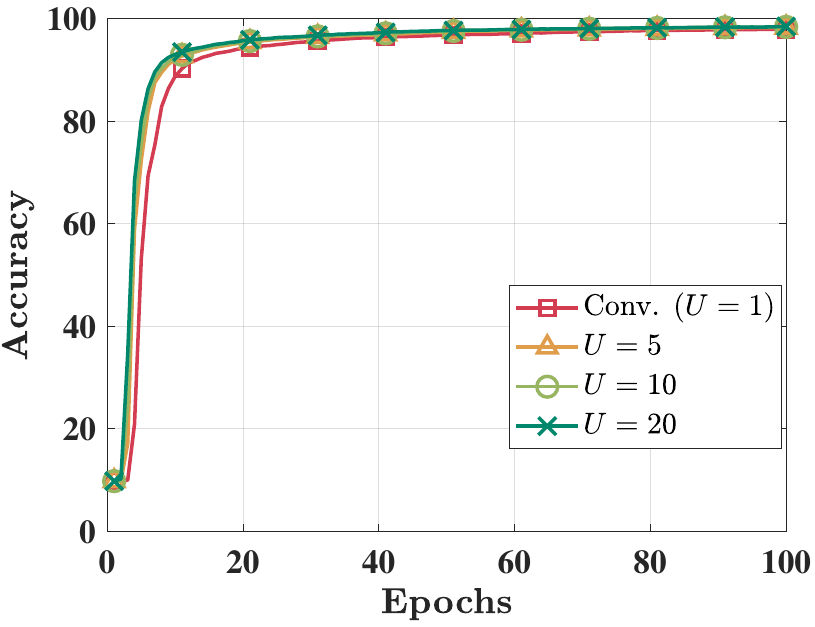}\label{fig:MNIST_P_010}} 
    \caption{Comparison of MNIST classifier performance for different numbers of antenna ports with a fixed transmit power budget.  \label{fig:MNIST_results_P}}
\end{figure*}

\begin{figure*}[t]
    \centering
    \subfloat[Conv. ($U=1$)]{\includegraphics[width=0.24\textwidth]{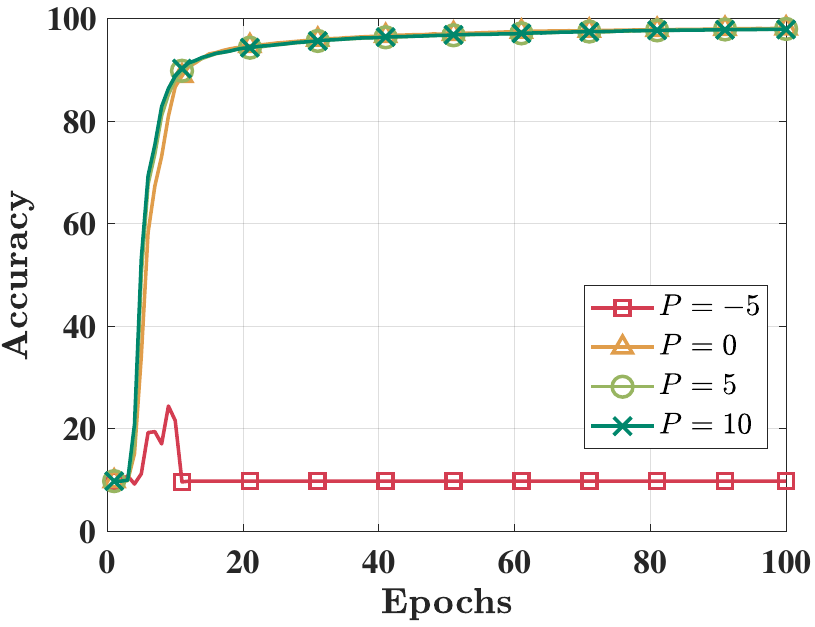}\label{fig:MNIST_U_001}} 
    \hspace{1pt}
    \subfloat[$U=5$]{\includegraphics[width=0.24\textwidth]{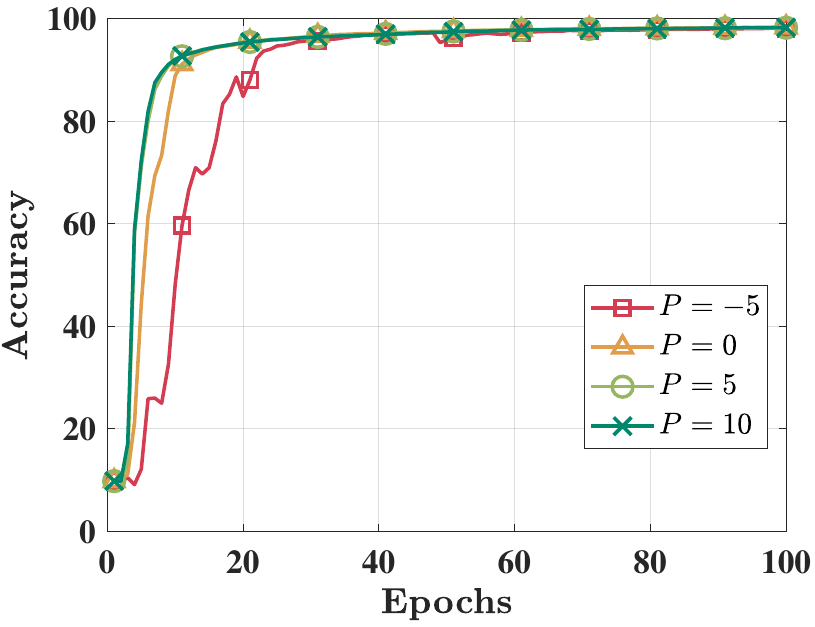}\label{fig:MNIST_U_005}} 
    \hspace{1pt}
    \subfloat[$U=10$]{\includegraphics[width=0.24\textwidth]{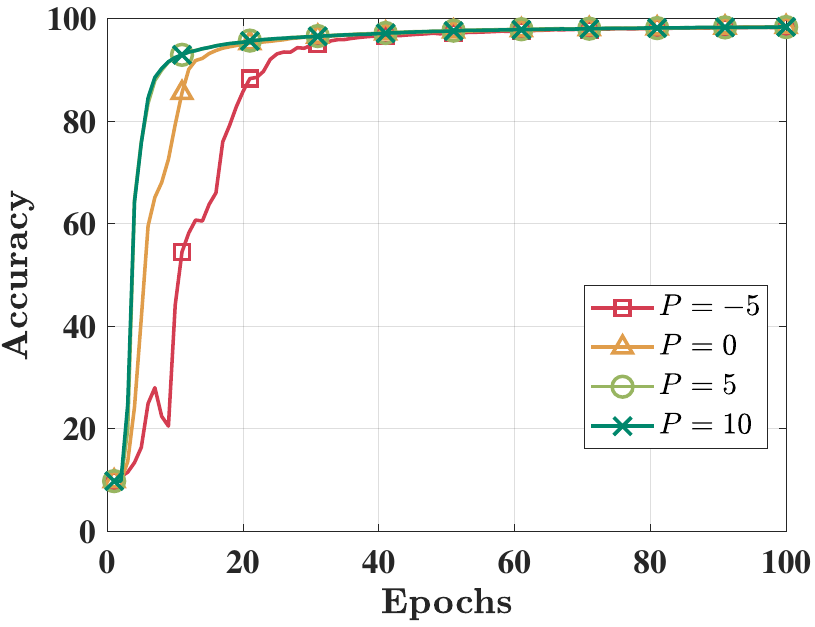}\label{fig:MNIST_U_010}} 
    \hspace{1pt}
    \subfloat[$U=20$]{\includegraphics[width=0.24\textwidth]{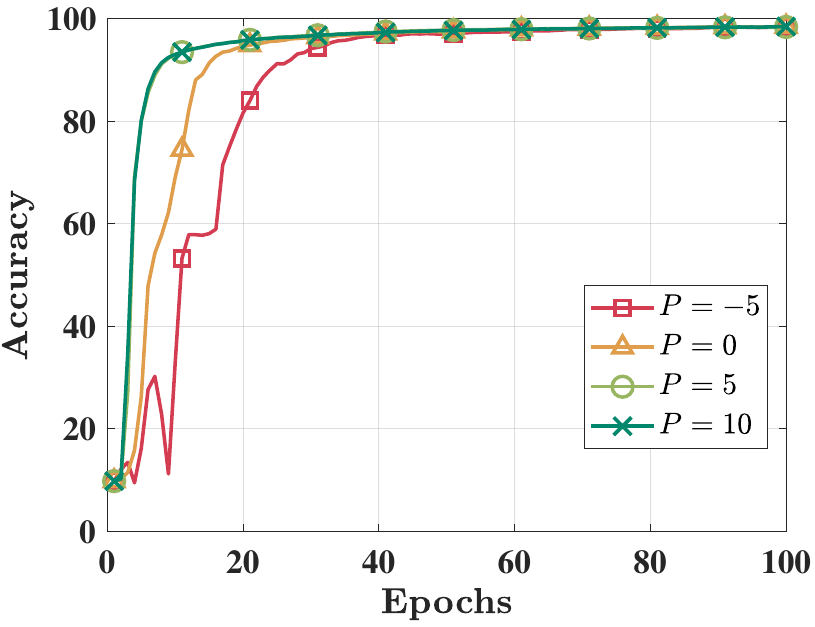}\label{fig:MNIST_U_020}} 
    \caption{Comparison of MNIST classifier performance for different transmit power budgets with the fixed number of antenna port. \label{fig:MNIST_results_U}}
\end{figure*}

\subsection{Simulation environment}
We developed an image classifier utilizing the MNIST \cite{MNIST} and CIFAR-10 \cite{CIFAR10} datasets, which are established benchmarks in image classification. The MNIST dataset includes 60,000 training and 10,000 test samples, while CIFAR-10 contains 50,000 training and 10,000 test samples across 10 classes.

Communication between devices and the server features wireless channels modeled with a width of $W=0.5$ and $F=64$ subcarriers per OFDM symbol. System noise is represented by a noise power of $N_0=-50$ dBm and a port selection threshold of $\psi=0.001$. The training employs mini-batch SGD with a batch size of $B=64$ and a learning rate of $\lambda=0.1$. Cross-entropy loss measures the discrepancy between predictions and actual labels, typical for multi-class classification.

\subsection{MNIST classifier}
The MNIST classifier employs a deep neural network featuring a mix of convolutional and fully connected layers. It includes three convolutional layers with 6, 16, and 120 filters respectively—the first two layers using $5\times5$ kernels and the last using a $4\times4$ kernel. A ReLu activation function follows each convolutional layer to enhance model expressiveness and mitigate vanishing gradients. The network also integrates two fully connected layers: the first with 84 neurons and the output layer with 10 neurons, matching the MNIST dataset's class count. The network comprises a total of $d=44,426$ trainable parameters. The setup involves $M=20$ devices, each holding $D_m=3,000$ unique, non-overlapping samples from the dataset.

\subsubsection{Various $U$}
Fig. \ref{fig:MNIST_results_P} shows training performance of classifiers trained by FAir-FL for the various number of antenna ports $U=1,5,10,20$, where FAir-FL when $U=1$ is the same as the conventional scheme without FAS. For all transmit power budgets, the conventional scheme represented by the square marked solid red line has worse performance than FAir-FL. Especially, for $P=-5$, the channel of a fixed single antenna has a much higher probability of experiencing deep fading which amplifies the power of the additive noise in the received GMU at server and consequently leads to training failure as shown in Fig. \ref{fig:MNIST_P_m05}. 

In addition, we observe that FAir-FL with multiple antenna ports successfully completes training in all four cases while achieving high performance in each case. This is because MNIST classification is relatively easy, and the training performance can be significantly improved even with a small number of antenna ports.

\subsubsection{Various $P$}
Fig. \ref{fig:MNIST_results_U} shows the training performance for various transmit power budgets $P=-5, 0, 5, 10$ with different numbers of antenna ports. Note that a higher transmit power budget with leads to better training performance in all scenarios. This is because higher transmit power increases $\gamma^{t,n}$ which results in the smaller noise level in the received GMU at the server. However, we observe that there is no performance gap between $P=5$ and $10$. This is because the noise level at the server becomes negligible when the device uses transmit power above 5, and thus the improvement in training performance is limited as transmit power increases.

\begin{figure*}[t]
    \centering
    \subfloat[$P=-2$ {[}dBm{]}]{\includegraphics[width=0.24\textwidth]{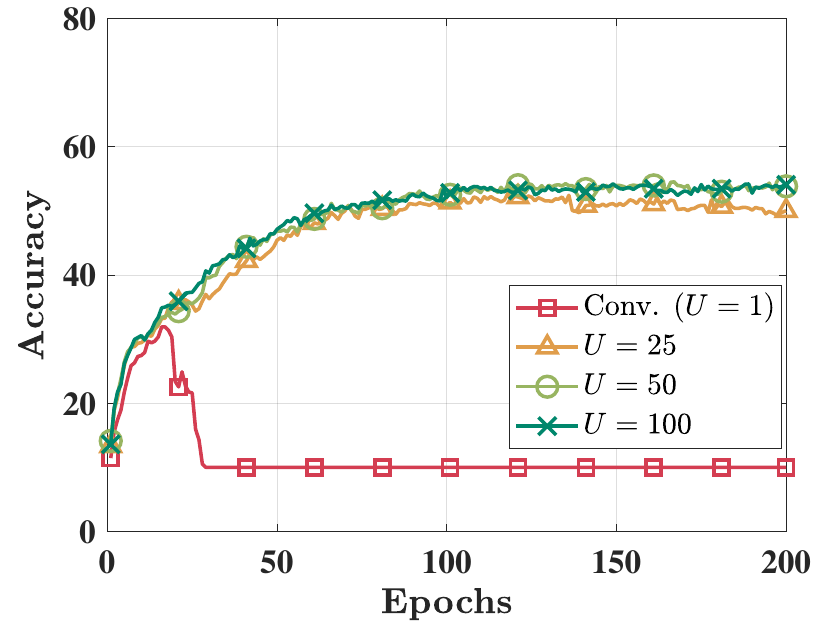}\label{fig:P_m02}} 
    \hspace{1pt}
    \subfloat[$P=0$ {[}dBm{]}]{\includegraphics[width=0.24\textwidth]{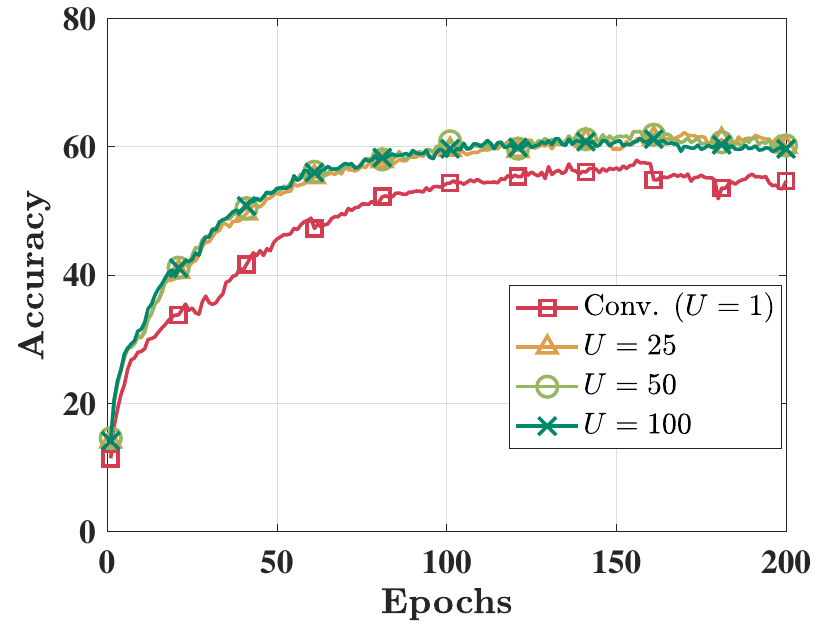}\label{fig:P_000}} 
    \hspace{1pt}
    \subfloat[$P=10$ {[}dBm{]}]{\includegraphics[width=0.24\textwidth]{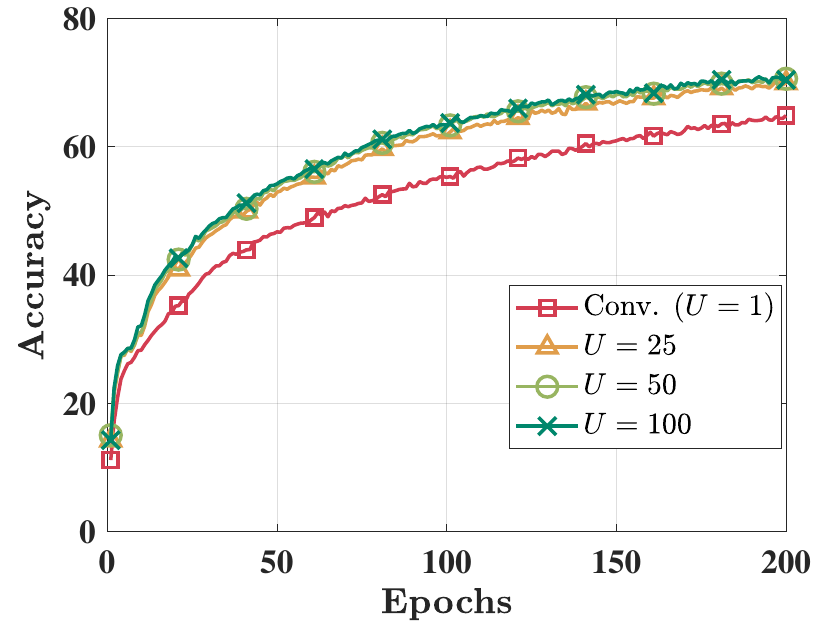}\label{fig:P_010}} 
    \hspace{1pt}
    \subfloat[$P=20$ {[}dBm{]}]{\includegraphics[width=0.24\textwidth]{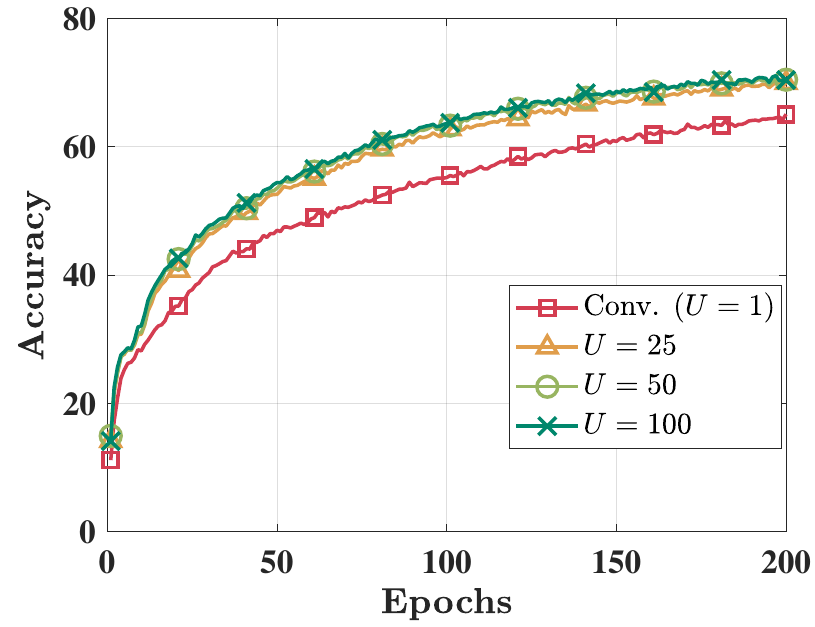}\label{fig:P_020}} 
    \caption{Comparison of CIFAR-10 classifier performance for different numbers of antenna ports with a fixed transmit power budget.  \label{fig:results_P}}
\end{figure*}

\begin{figure*}[t]
    \centering
    \subfloat[Conv. ($U=1$)]{\includegraphics[width=0.24\textwidth]{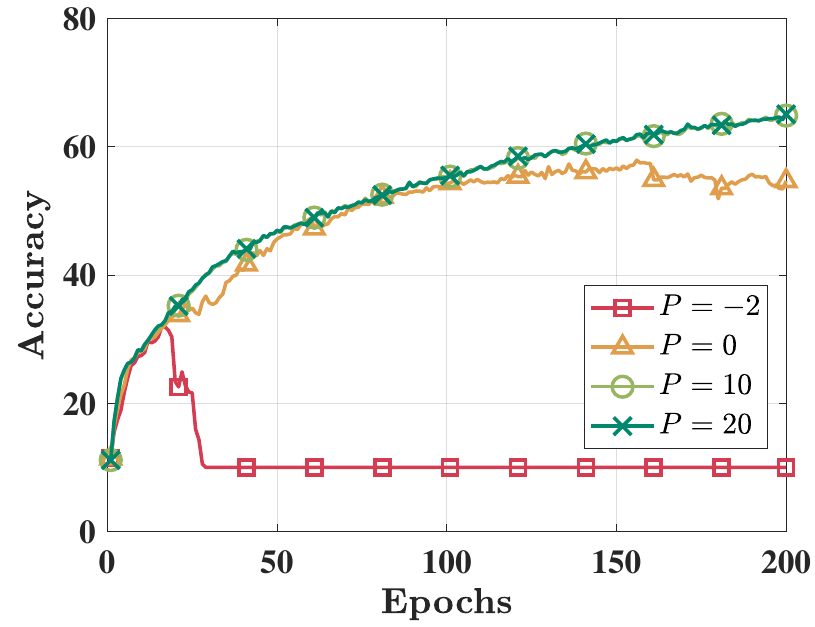}\label{fig:U_001}} 
    \hspace{1pt}
    \subfloat[$U=25$]{\includegraphics[width=0.24\textwidth]{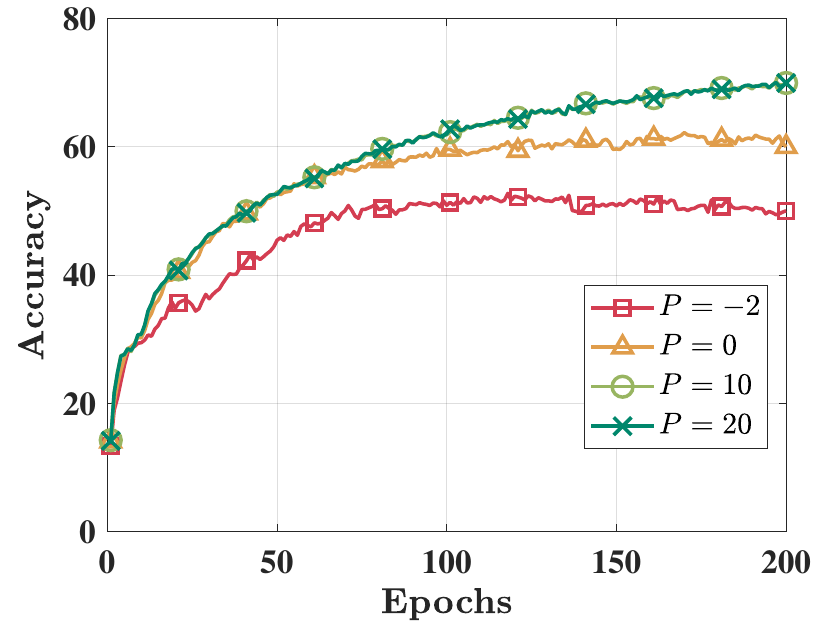}\label{fig:U_025}} 
    \hspace{1pt}
    \subfloat[$U=50$]{\includegraphics[width=0.24\textwidth]{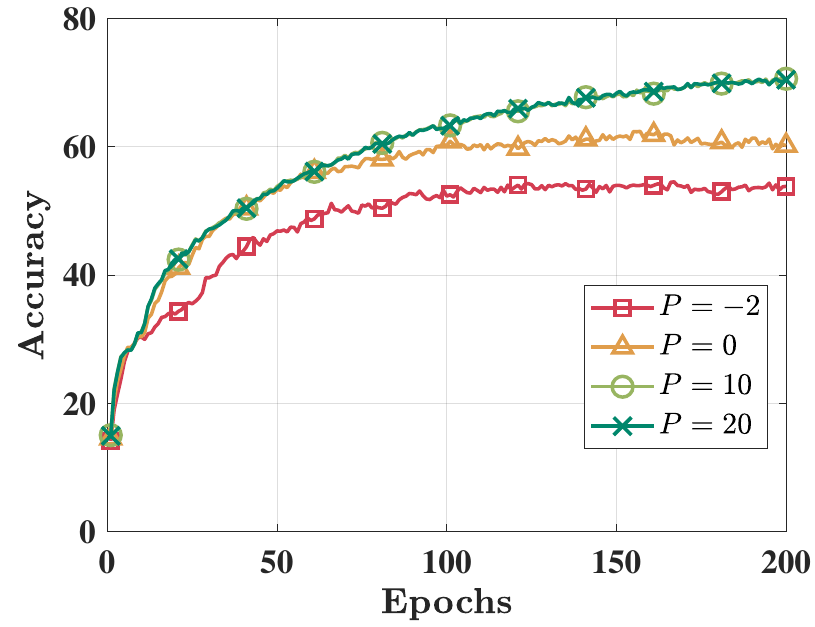}\label{fig:U_050}} 
    \hspace{1pt}
    \subfloat[$U=100$]{\includegraphics[width=0.24\textwidth]{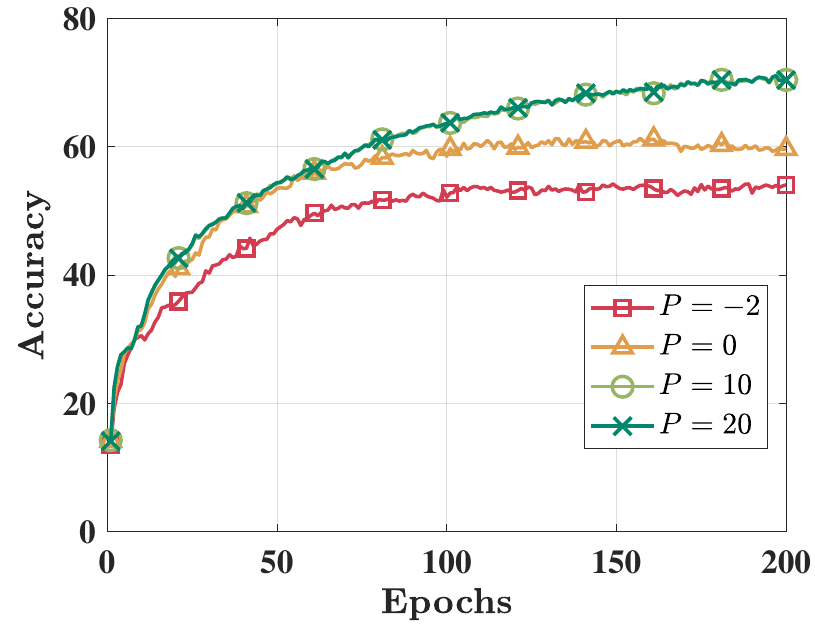}\label{fig:U_100}} 
    \caption{Comparison of CIFAR-10 classifier performance for different transmit power budgets with the fixed number of antenna port.  \label{fig:results_U}}
\end{figure*}

\subsection{CIFAR-10 classifier}
The CIFAR-10 classifier is designed as a deep neural network with a combination of convolutional and fully connected layers. Specifically, it comprises three convolutional layers followed by four fully connected layers. The convolutional layers utilize 16, 32, and 64 filters, respectively, where each filter has a kernel size of $3\times 3$. Each convolutional layer is followed by a nonlinear activation function ReLu. Following the convolutional layers, we implement two fully connected layers, where the first layer consists of 500 neurons, while the output layer contains 10 neurons corresponding to the 10 classes in CIFAR-10. The total number of trainable parameters in our network amounts to $d=541,094$.
We consider a total of $M=10$ devices where device contains $D_m=5,000$ unique samples where the data distribution is non-overlapping across different devices.

\subsubsection{Various $U$}
Fig. \ref{fig:results_P} illustrates the training performance of classifiers using FAir-FL across varying antenna ports ($U=1,25,50,100$) and transmit power budgets. The conventional scheme with a single antenna underperforms significantly compared to FAir-FL's multiple antenna configurations. Specifically, a single fixed antenna is susceptible to deep fading, negatively impacting performance at lower powers such as $P=-2$, where it fails to counteract the amplified noise power at the server, leading to training failures as depicted in Fig. \ref{fig:P_m02}.

FAir-FL with multiple antenna ports consistently completes training successfully across all configurations, demonstrating enhanced performance with an increase in antenna ports. At lower transmit powers ($P=-2$ and $0$), performance is compromised by elevated noise levels at the received GMU at the server. Devices mitigate this by selecting antenna ports that optimize $\gamma^{t,n}$, thereby controlling the additive noise level.

At higher powers ($P=10$ and $20$), training performance is influenced by the ability to transmit more parameters to the server. Devices enhance this capability by choosing antenna ports that maximize channel gains, increasing the number of transmitted parameters and thus improving training outcomes.

\subsubsection{Simulation results for various $P$}
Fig. \ref{fig:results_U} presents the training performance for various transmit power budgets $P=-2, 0, 10, 20$ with different numbers of antenna ports. As expected, a higher transmit power brings better training performance in all cases. This is because a higher transmit power budget increases $\gamma^{t,n}$ which results in a lower additive noise power in the received GMU at the server. Additionally, we observe that the performance gap between $P=10$ and $20$  is small. Similar to previous MNIST case, this is because when the transmit power budget exceeds 10, the noise level at the server becomes negligible, and thus further improvements in training are limited.

\section{Conclusion}
In this paper, we introduced the innovative FAir-FL algorithm, which optimizes antenna port selection based on the noise power levels at the parameter server. This approach enhances training stability under high noise conditions and improves performance in low noise scenarios. We have mathematically analyzed the convergence of FAir-FL under non-convex loss functions. Our simulation results confirm that using fluid antennas significantly outperforms fixed antennas in over-the-air FL settings, validating the efficacy of our proposed method.


%

\appendices
\section{Useful Lemmas}
\begin{lemma} \label{lem:bounded}
    For \textbf{Algorithm 1}, 
    \begin{align}
        \bbE[\lVert \bfg_m^{t,e} - \nabla L_m(\bsth_m^{t,e-1}) \rVert^2] \le \sigma^2. \nn
    \end{align}
\end{lemma}
\begin{proof} 
    It can be derived as follows.
    \begin{align}
        &\bbE[\lVert \bfg_m^{t,e} - \nabla L_m(\bsth_m^{t,e-1}) \rVert^2]
        \nn \\
        &=\bbE\bigg[\bigg\lVert \frac{1}{B}\sum_{\bfs\in\calB_{m}^{t,e}}\big(\nabla l(\bfs;\bsth_m^{t,e-1}) - \nabla L_m(\bsth_m^{t,e-1})\big)\bigg\rVert^2\bigg] \nn \\
        &\le \frac{1}{B}\sum_{\bfs\in\calB_m^{t,e}} \bbE[\lVert\nabla l(\bfs;\bsth_m^{t,e-1}) - \nabla L_m(\bsth_m^{t,e-1})\rVert^2] \nn \\
        &\lr\le{lem_sigma} \sigma^2, \nn
    \end{align}
    where \eqref{lem_sigma} holds due to \eqref{eq:boundness1} in \textbf{Assumption 2}.
\end{proof}

\begin{lemma} \label{lem:sumsum_g}
    For \textbf{Algorithm 1}, if $\lambda\le \frac{1}{2\sqrt{2}\kappa E}$, then,
    \begin{align}
        &\sum_{m\in\calM}\sum_{e=1}^{E}\bbE[\lVert \bfg_m^{t,e} \rVert^2] \nn\\
        &\qquad\le  8 M E \sigma^2 + 8M^2E\delta^2 + 8ME\cdot \bbE[\lVert \nabla L(\bsth^t) \rVert^2]. \nn  
    \end{align}
\end{lemma}
\begin{proof}
    For ease of notation, we omit the super script $t$.
    
    Firstly, $\bbE[\lVert \bfg_m^{e} \rVert^2]$ can be bounded as in \eqref{lem_g_01}-\eqref{lem_g_02} where \eqref{sumsum_g_01} holds due to \eqref{eq:smooth} in \textbf{Assumption 1}, \textbf{Lemma \ref{lem:bounded}} and \eqref{eq:diversity} in \textbf{Assumption 4}.
\begin{figure*}
    \begin{align}
    &\bbE[\lVert \bfg_m^{e} \rVert^2]
    =\bbE[\lVert \bfg_m^{e} - \nabla L_m(\bsth_m^{e-1}) + \nabla L_m(\bsth_m^{e-1}) - \nabla L_m(\bsth) + \nabla L_m(\bsth) - \nabla L(\bsth) + \nabla L(\bsth)  \rVert^2] \label{lem_g_01}  \\
    &\le 4 \bbE[\lVert \bfg_m^{e} - \nabla L_m(\bsth_m^{e-1}) \rVert^2+\lVert \nabla L_m(\bsth_m^{e-1}) - \nabla L_m(\bsth) \rVert^2+\lVert \nabla L_m(\bsth) - \nabla L(\bsth)  \rVert^2+\lVert \nabla L(\bsth)\rVert^2] \\
    &\lr\le{sumsum_g_01} 4\sigma^2 + 4\kappa^2 \bbE[\lVert \bsth_m^{e-1} - \bsth \rVert^2] + 4 M \delta^2 + 4 \bbE[\lVert \nabla L(\bsth) \rVert^2] \\
    &= 4\sigma^2 + 4\kappa^2 \bbE\bigg[\bigg\lVert \sum_{\epsilon=2}^{e}\lambda\bfg_{m}^{\epsilon-1}\bigg\rVert^2\bigg] + 4 M \delta^2 + 4 \bbE[\lVert \nabla L(\bsth) \rVert^2] \\
    &\le 4\sigma^2 + 4\lambda^2\kappa^2  E \sum_{e=1}^{E}\bbE[\lVert \bfg_{m}^{e}\rVert^2] + 4 M \delta^2 + 4 \bbE[\lVert \nabla L(\bsth) \rVert^2]. \label{lem_g_02}
    \end{align}
    \hrule
\end{figure*}
Then, by summing \eqref{lem_g_02} from $e=1$ to $E$ and for $m\in\calM$, it can be described as
\begin{align}
    &\sum_{m\in\calM}\sum_{e=1}^{E}\bbE[\lVert \bfg_m^{e} \rVert^2] \nn \\ 
    &\qquad\le 4 M E \sigma^2 + 4 \lambda^2\kappa^2  E^2 \sum_{m\in\calM}\sum_{e=1}^{E}\bbE[\lVert \bfg_m^{e} \rVert^2] \nn \\
    &\qquad\qquad+ 4M^2E\delta^2 + 4ME\cdot \bbE[\lVert \nabla L(\bsth) \rVert^2], \nn
\end{align}
    and it can be rewritten as
    \begin{align}
        &(1-4\lambda^2\kappa^2 E^2) \sum_{m\in\calM}\sum_{e=1}^{E}\bbE[\lVert \bfg_m^{e} \rVert^2] \nn \\
        &\le  4 M E \sigma^2 + 4M^2E\delta^2 + 4ME\cdot \bbE[\lVert \nabla L(\bsth) \rVert^2]. \label{lem_g_03}
    \end{align}
Since $\lambda<\frac{1}{2\sqrt{2}\kappa E}$ implies $1-4\lambda^2\kappa^2 E^2 > \frac{1}{2}$, \eqref{lem_g_03} can be reduced to the desired result.
\end{proof}

\section{Proof of Theorem \ref{lemma01}}
Based on $\kappa$-smoothness of $L_m$'s in \eqref{eq:smooth2} in \textbf{Assumption 1} for $\bfx=\bsth^{t}$ and $ \bfy=\bsth^{t+1}$, we can obtain the following inequality.
\begin{align}
    &\bbE[L(\bsth^{t+1})] - \bbE[L(\bsth^{t})] \nn \\
    &\le 
     \underbrace{\bbE[\nabla L(\bsth^t)^{\mathsf{T}}(\bsth^{t+1} - \bsth^t)]}_{\triangleq\text{(A.1)}} 
    + \frac{\kappa}{2}\cdot \underbrace{\bbE[\lVert\bsth^{t+1} - \bsth^t\rVert^2]}_{\triangleq\text{(A.2)}}. \label{pflem:eq01}
\end{align}
For ease of notation, we omit the superscript $t$ again.

Firstly, we will derive the upper bound of (A.1). It can be rewritten as 
\begin{align}
    &\text{(A.1)}\lr={appB01} \bbE\bigg[\nabla L(\bsth)^{\mathsf{T}}
    \sum_{m\in\calM}\rho_m\Delta_m\odot\mathbbm{1}_{\calH(\tau)}(\tilde{\bfh}_{m})\bigg] \nn \\ 
    &\lr={appB02} -\lambda e^{-\tau}  \bbE\bigg[\nabla L(\bsth)^{\mathsf{T}}
    \sum_{m\in\calM}\rho_m\sum_{e=1}^E \bfg_{m}^{e}\bigg] \nn \\ 
    &\lr{=}{appB03} -\lambda e^{-\tau} \sum_{e=1}^E \bbE\bigg[\nabla L(\bsth)^{\mathsf{T}}
    \sum_{m\in\calM}\rho_m\nabla L_m(\bsth_m^{e-1})\bigg] \nn \\ 
    &\lr{=}{appB04} -\frac{\lambda e^{-\tau}}{2} \sum_{e=1}^{E}\bigg( \bbE[\lVert\nabla L(\bsth)\rVert^2] \nn \\
    &\quad+ \bbE\bigg[\bigg\lVert \sum_{m\in\calM}\rho_m\nabla L_m(\bsth_{m}^{e-1})\bigg\rVert^2 \bigg] \nn \\ 
    &\quad -  \underbrace{\bbE\bigg[\bigg\lVert \nabla L(\bsth) - \sum_{m\in\calM} \rho_m \nabla L_m(\bsth_{m}^{e-1}) \bigg\rVert^2\bigg]}_{\triangleq\text{(B)}} \bigg)
   \label{ineq:partA}
\end{align}
for $\tilde{\bfh}_m = [(\tilde{\bfh}_{m}^{1})^{\mathsf{T}},\cdots, (\tilde{\bfh}_{m}^{N})^{\mathsf{T}}]^{\mathsf{T}}$  
where \eqref{appB01} holds due to independence of $\bfw^n$ and $\bbE[\bfw^n]=\mathbf{0}_F$,
\eqref{appB02} holds due to  
\begin{align}
    \bbE[\mathbbm{1}_{\calH(\tau)}(\tilde{h}_{m,f}^n)]
    & = \frac{1}{U}\sum_{u=1}^{U}\bbE[\mathbbm{1}_{\calH(\tau)}(h_{m,u,f}^{n})] \nn \\
    &=e^{-\tau} \nn
\end{align}
when port is selected uniformly at random, \eqref{appB03} holds due to 
\begin{align}
    \bbE[\bfg_m^{e}] &= \frac{1}{|\calB_m^{e}|}\sum_{\bfs\in\calB_m^{e-1}}\bbE[ \nabla l(\bfs;\bsth_m^{e-1}) ] \nn \\
    &=\nabla L_m(\bsth_m^{e-1})  \nn
\end{align}
 and \eqref{appB04} holds due to $\bfa^{\mathsf{T}} \bfb = \frac{1}{2}(\lVert \bfa \rVert^2 + \lVert \bfb \rVert^2 - \lVert \bfa - \bfb \rVert^2)$ for any $\bfa$ and $\bfb$.
Furthermore, (B) can be bounded by
\begin{align}
    \text{(B)}&=\bbE\bigg[\bigg\lVert \sum_{m\in\calM} \rho_m (\nabla L_m(\bsth) -  \nabla L_m(\bsth_{m}^{e-1})) \bigg\rVert^2\bigg] \nn \\
    &\le \sum_{m\in\calM} \rho_m\bbE[\lVert  \nabla L_m(\bsth) -  \nabla L_m(\bsth_{m}^{e-1}) \rVert^2] \nn \\ 
    &\lr{\le}{appB05} \kappa^2 \sum_{m\in\calM} \rho_m  \cdot  {\bbE[\lVert\bsth^t-\bsth_m^{e-1}\rVert^2]}  \nn\\
    &= \lambda^2\kappa^2   \sum_{m\in\calM} \rho_m \bbE\bigg[\bigg\lVert \sum_{\epsilon=2}^{e}\bfg_m^{\epsilon-1}\bigg\rVert^2\bigg]    \nn \\
    &\le \lambda^2 \kappa^2  E \sum_{m\in\calM} \sum_{e=1}^{E} \bbE[\lVert \bfg_m^{e} \rVert^2]
    \label{eq:B_result}
\end{align}
where \eqref{appB05} holds due to $\kappa$-smoothness of $L_m$'s in \eqref{eq:smooth} in \textbf{Assumption 1}.
Then, by applying \textbf{Lemma \ref{lem:sumsum_g}} to \eqref{eq:B_result}, we can achieve the following result.
\begin{align}
    &\text{(B)} 
    \le 8 \lambda^2 \kappa^2  M E^2 \sigma^2 + 8\lambda^2 \kappa^2M^2E^2\delta^2 \nn\\
    &\qquad\qquad+ 8\lambda^2 \kappa^2ME^2\cdot \bbE[\lVert \nabla L(\bsth) \rVert^2]. \label{B_result}
\end{align}
By substituting \eqref{B_result} to \eqref{ineq:partA}, we can obtain
\begin{align}
    &\text{(A.1)} \le -\frac{\lambda e^{-\tau} E}{2} (1 - 8\lambda^2\kappa^2 M E^2) \bbE[\lVert \nabla L(\bsth)\rVert^2] \nn \\
    &\qquad - \frac{ \lambda e^{-\tau} }{2} \sum_{e=1}^E  \bbE\bigg[\bigg\lVert \sum_{m\in\calM}\rho_m\nabla L_m(\bsth_{m}^{e-1})\bigg\rVert^2\bigg] \nn \\
    &\qquad + 4 \lambda^3 e^{-\tau} \kappa^2 M E^3 \sigma^2 + 4 \lambda^3 e^{-\tau} \kappa^2 M^2 E^3 \delta^2 \nn \\
    &\le -\frac{\lambda e^{-\tau} E}{4}  \bbE[\lVert \nabla L(\bsth)\rVert^2] \nn \\
    &\qquad + 4 \lambda^3 e^{-\tau} \kappa^2 M E^3 \sigma^2 + 4 \lambda^3 e^{-\tau} \kappa^2 M^2 E^3 \delta^2,
    \label{A1_result}  
\end{align}
since $\lambda<\frac{e^{-\tau}}{32\kappa M E}<\frac{1}{4\kappa\sqrt{M}E}$ implies $1-8\lambda^2 \kappa^2 M E^2 < \frac{1}{2}$. 

Secondly, (A.2) can be derived as 
\begin{align}
    &\text{(A.2)}= \bbE\bigg[\bigg\lVert \sum_{m\in\calM} \rho_m\Delta_m\odot\mathbbm{1}_{\calH(\tau)}(\bfh_{m}) +\tilde{\bfw} \bigg\rVert^2\bigg] \nn \\
    &\lr{=}{appB06} \underbrace{\bbE\bigg[\bigg\lVert \sum_{m\in\calM} \rho_m\Delta_m\odot\mathbbm{1}_{\calH(\tau)}(\bfh_{m}) \bigg\rVert^2\bigg]}_{\triangleq\text{(C.1)}} + \underbrace{\bbE[\lVert \tilde{\bfw} \rVert^2]}_{\triangleq\text{(C.2)}} \label{ineq:partB}
\end{align}
where \eqref{appB06} comes from independence of $\bfw^n$ and $\bbE[\bfw^n]=\mathbf{0}_F$. Then, (C.1) can be bounded by
\begin{align}
    \text{(C.1)}
    &\le \lambda^2 E\sum_{m\in\calM} \sum_{e=1}^{E}\bbE[\lVert  \bfg_m^{e}  \odot\mathbbm{1}_{\calH(\tau)}(\bfh_{m}) \rVert^2]  \nn \\
    &\lr\le{appB07} \lambda^2 E\sum_{m\in\calM} \sum_{e=1}^{E}\bbE[\lVert  \bfg_m^{e}  \rVert^2] \nn \\
    &\lr\le{appB08} 8\lambda^2 M E^2 \sigma^2 + 8\lambda^2 M^2E^2\delta^2 \nn \\
    & \qquad \qquad + 8\lambda^2 ME^2 \cdot \bbE[\lVert \nabla L(\bsth) \rVert^2], \label{C1result}
\end{align}
where \eqref{appB07} and \eqref{appB08} hold due to $\lVert \bfa \odot \mathbbm{1}_{\calA}(\bfb)  \rVert \le \lVert \bfa \rVert$ for any vector $\bfa, \bfb$ and set $\calA$ and \textbf{Lemma \ref{lem:sumsum_g}}, respectively. Also, (C.2) can be rewritten as
\begin{align}
    \text{(C.2)}
    &= \sum_{n=1}^{N}\bbE[\lVert\tilde{\bfw}^{n}\rVert^2]    
    \nn \\
    &=\sum_{n=1}^{N} \bbE[\bbE[\lVert\tilde{\bfw}^{n}\rVert^2|\gamma^{n}]] \nn \\
    &=\sum_{n=1}^{N} N_0F\cdot \bbE[(\gamma^{n})^{-2}] \label{C2result} 
\end{align}
Therefore, by substituting \eqref{C1result} and \eqref{C2result} to  \eqref{ineq:partB}, we can derive the following upper bound of (A.2).
\begin{align}
    &\text{(A.2)}
    \le 8\lambda^2 M E^2 \sigma^2 + 8 \lambda^2 M^2 E^2 \delta^2 \nn \\
    &+ 8\lambda^2 M E^2 \cdot \bbE[\lVert \nabla L(\bsth) \rVert^2]
    +\sum_{n=1}^{N} \bbE\Bigg[ \frac{N_0 F}{(\gamma^{n})^{-2}} \Bigg] \label{A2result}
\end{align}
To avoid confusion, we will use the superscript $t$
again. By substituting \eqref{A1_result} and \eqref{A2result} to \eqref{pflem:eq01}, we can achieve the following inequality.
\begin{align}
    &\bbE[L(\bsth^{t+1})] - \bbE[L(\bsth^t)] \nn \\
    &\le -\frac{\lambda  E}{4}( e^{-\tau} - 16 \lambda \kappa M E)  \bbE[\lVert \nabla L(\bsth^t)\rVert^2] \nn \\
    &+ 4 \lambda^3 e^{-\tau} \kappa^2 M E^3 \sigma^2 + 4 \lambda^3 e^{-\tau} \kappa^2 M^2 E^3 \delta^2 \nn \\
    &+4\lambda^2 \kappa M E^2 \sigma^2 + 4 \lambda^2 \kappa M^2 E^2 \delta^2 +\sum_{n=1}^{N} \bbE\Bigg[ \frac{\kappa N_0 F}{2(\gamma^{t,n})^{-2}} \Bigg] \nn \\
     &\le -\frac{\lambda e^{-\tau}  E}{8}  \bbE[\lVert \nabla L(\bsth^t)\rVert^2] \nn \\
    &+ 4 \lambda^3 e^{-\tau} \kappa^2 M E^3 \sigma^2 + 4 \lambda^3 e^{-\tau} \kappa^2 M^2 E^3 \delta^2 \nn \\
    &+4\lambda^2 \kappa M E^2 \sigma^2 + 4 \lambda^2 \kappa M^2 E^2 \delta^2 +\sum_{n=1}^{N} \bbE\Bigg[ \frac{\kappa N_0 F}{2(\gamma^{t,n})^{-2}} \Bigg] \nn
\end{align}
where $\lambda < \frac{e^{-\tau}}{32 \kappa M E }$ implies $e^{-\tau} - 16\lambda \kappa ME > \frac{e^{-\tau}}{2}$.
Reformulating the above result, we can obtain 
\begin{align}
    &\frac{\lambda e^{-\tau} E}{8} \bbE[\lVert \nabla L(\bsth^t)\rVert^2] \le \bbE[L(\bsth^t)] - \bbE[L(\bsth^{t+1})] \nn \\ 
    &+ 4 \lambda^3 e^{-\tau} \kappa^2 M E^3 \sigma^2 + 4 \lambda^3 e^{-\tau} \kappa^2 M^2 E^3 \delta^2 \nn \\
    &+4\lambda^2 \kappa M E^2 \sigma^2 + 4 \lambda^2 \kappa M^2 E^2 \delta^2 +\sum_{n=1}^{N} \bbE\Bigg[ \frac{\kappa N_0 F}{2(\gamma^{t,n})^{-2}} \Bigg]\label{reform_01}
\end{align}
By summing \eqref{reform_01} from $t=1$ to $t=T$ and divide the both side with $\lambda e^{-\tau} E T/ 8$, we can achieve the desired result as
\begin{align}
    &\bbE\bigg[\frac{1}{T}\sum_{t=1}^{T}\lVert \nabla L(\bsth^{t}) \rVert^2 \bigg] 
    \lr\le{lr12} \frac{8(L(\bsth^1) - L(\bsth^*))}{\lambda e^{-\tau} E T} \nn \\
    &\qquad+32  \lambda^2  \kappa^2 M E^2 \sigma^2 + 32 \lambda^2  \kappa^2 M^2 E^2 \delta^2 \nn \\
    &\qquad+\frac{32\lambda \kappa M E \sigma^2}{e^{-\tau}} + \frac{32\lambda \kappa M^2 E \delta^2}{e^{-\tau}} \nn \\
    &\qquad+\frac{4\kappa N_0  F}{\lambda e^{-\tau} E T}\sum_{t=1}^{T}\sum_{n=1}^{N}\bbE[ (\gamma^{t,n})^{-2}], \nn
\end{align}
where \eqref{lr12} holds since $L(\bsth^*)$ is the minimum.

\section{Proof of Theorem \ref{lemma02}}

For ease of notation, we omit the superscripts $t$ and $n$.
The expectation of $\bbE[\bar{\gamma}^{-2}]$ can be bounded by
\begin{align}
    \bbE[\bar{\gamma}^{-2}]
    &\lr\le{def_gamma} \bbE[\bar{\gamma}_2^{-2}]\nn \\
    &= \bbE\Bigg[ \frac{\beta^2 \bar{\zeta}^{-2}}{P} \Bigg] \nn \\
    &= \frac{1}{P}\cdot \underbrace{\bbE[\beta^2]}_{\triangleq\text{(D.1)}}\underbrace{\bbE[\bar{\zeta}^{-2}]}_{\triangleq\text{(D.2)}} \label{result_L1L2}
\end{align}
for 
\begin{align}
    \bar{\zeta} &= \min_{m}\Bigg(\sum_{f=1}^{F}\lvert 
    \tilde{h}_{m,f} \rvert^{-2} \mathbbm{1}_{\calH(\tau)}(\tilde{h}_{m,f}) \Bigg)^{-\frac{1}{2}} \nn
\end{align}
with selected port in \eqref{eq:selection02} where \eqref{def_gamma} holds due to the definition of selection rule.

Firstly, we will find a upper bound of (D.1). Note that  by the definition of $\beta$ in \eqref{def:beta}, we can obtain 
\begin{align}
    \text{(D.1)} &= \bbE\bigg[\max_{m,j} \lvert\delta_{m,j}\rvert^2\bigg] \nn \\
    &\lr\le{lemlr01} \bbE\bigg[\max_m \lVert \Delta_m \rVert^2\bigg] \nn \\
    &= \lambda^2\bbE\Bigg[\max_m \Bigg\lVert \sum_{e=1}^{E}  \bfg_m^{e} \Bigg\rVert^2\Bigg] \nn \\
    &\lr\le{lemlr02} \lambda^2 \sum_{m\in\calM} \bbE\Bigg[ \Bigg\lVert \sum_{e=1}^{E}  \bfg_m^{e} \Bigg\rVert^2 \Bigg] \nn \\
    &\lr\le{lemlr03} \lambda^2 E \sum_{m\in\calM}\sum_{e=1}^{E} \bbE[ \lVert \bfg_m^{e} \rVert^2] \nn \\
    &\lr\le{lemlr04} 8\lambda^2 M E^2 \sigma^2 + 8\lambda^2 M^2E^2\delta^2 \nn \\
    &\qquad\qquad+ 8\lambda^2 ME^2 \cdot \bbE[\lVert \nabla L(\bsth) \rVert^2], \label{result_L_1}
\end{align}
where \eqref{lemlr01} and \eqref{lemlr02} hold from the fact that $\max_m \lVert \bfa_m \rVert^2 \le \sum_m\lVert \bfa_m \rVert^2$ for any $\bfa_m$, \eqref{lemlr03} holds due to Cauchy-Scharwz inequality and \eqref{lemlr04} holds from Lemma \ref{lem:sumsum_g}. 

Secondly, we define a new random variable $\zeta$ as
\begin{align}
    \zeta = \min_m \bigg( \tau^{-2} \max_{u} \nu_{m,u} \bigg)^{-\frac{1}{2}}, \nn
\end{align}
where $\nu_{m,u}=\mathbf{1}_F^{\mathsf{T}} \mathbbm{1}_{\calH(\tau)}(\bfh_{m,u})$. For any device $m$, 
\begin{align}
    \sum_{f=1}^{F} \lvert \tilde{h}_{m,f} \rvert^{-2} \mathbbm{1}_{\calH(\tau)}(\tilde{h}_{m,f}) 
    &\le \tau^{-2}  \nu_{m,u_{m}} \nn \\
    &= \tau^{-2} \max_u \nu_{m,u}. \nn
\end{align}
Therefore, we can obtain that $\zeta\le \bar{\zeta}$.
Before addressing (D.2), we can derive the cumulative distribution function of $\zeta^{-2}$ for $x>0$ as follows.
\begin{align}
    &\Pr[\zeta^{-2}\le x] = \Pr\bigg[ \zeta \ge \frac{1}{\sqrt{x}} \bigg] \nn\\
    &=\Pr\Bigg[ \min_m \Big( \tau^{-2}\max_u \nu_{m,u} \Big)^{-\frac{1}{2}} \ge\frac{1}{\sqrt{x}} \Bigg] \nn \\
    &=\Pr\Bigg[ \bigcap_m \Big( \tau^{-2}\max_u  \nu_{m,u} \Big)^{-\frac{1}{2}} \ge\frac{1}{\sqrt{x}} \Bigg] \nn \\
    &=\prod_m \Pr\big[ \max_u\nu_{m,u} \le \tau^2 x \big] \nn  \\ 
    &= \prod_m \underbrace{\Pr \bigg[\bigcap_u \nu_{m,u} \le \tau^2 x \bigg]}_{\triangleq\text{(E)}}. \nn
\end{align}
Then, (E) has a lower bound as 
\begin{align}
    &\text{(E)} \nn \\
    &= \Pr \Bigg[ \bigcap_u \mathbf{1}_F^{\mathsf{T}} \mathbbm{1}_{\calH(\tau^2)}((1-\mu_u^2)\bfz_{m,u} + \mu_u^2 \bfz_{m,1}) \le \tau^2 x \Bigg]\nn\\
    &\ge\Pr \Bigg[ \bigcap_u \mathbf{1}_F^{\mathsf{T}} \big(\mathbbm{1}_{\calH(p_{u})}(\bfz_{m,u}) + \mathbbm{1}_{\calH(q_u)}(\bfz_{m,1})\big) \le \tau^2 x \Bigg] \nn \\
    &\ge \Pr\Bigg[ \bigcap_u \mathbf{1}_F^\mathsf{T} \mathbbm{1}_{\calH(r)}(\bfz_{m,u}) \le \tau^2 x -F \Bigg] \nn \\
    &= \prod_u \underbrace{\Pr\Big[ \mathbf{1}_F^\mathsf{T} \mathbbm{1}_{\calH(r)}(\bfz_{m,u}) \le \tau^2 x -F \Big]}_{\triangleq \text{(F)}} \nn
\end{align}
where $\bfz_{m,u}=[z_{m,u,1},\cdots,z_{m,u,F}]^{\mathsf{T}}$ with $z_{m,u,f} = \frac{1}{2}(\lvert a_{m,u,f} \rvert^2 + \lvert b_{m,u,f} \rvert^2)$ following independent exponential distribution with exp(1), and
$p_u=\frac{\tau^2}{2(1-\mu_u^2)}$, $q_u=\frac{\tau^2}{2\mu_u^2}$ and $r = \frac{\tau^2}{2(1-\mu_{\min}^2)}$ for $\mu_{\min} = \min_u \mu_u$.
Note that $\mathbf{1}_F^{\mathsf{T}} \mathbbm{1}_{\calH(r)}(\bfz_{m,u})=\sum_{f} \mathbbm{1}_{\calH(r)}(z_{m,u,f})$ follows the binomial distribution with $\calB(F,\hat{p})$ with $\hat{p}=e^{-r}$. So, (F) can be described as
\begin{align}
    \text{(F)}= \sum_{f=0}^{\lfloor \tau^2 x - F \rfloor}  \binom{F}{f} \hat{p}^f (1-\hat{p})^{F-f}. \nn
\end{align}
Therefore, for $x$ with $ F\tau^{-2}\le x \le 2F\tau^{-2}$,
\begin{align}
    &\Pr[\zeta^{-2}\le x] \nn \\
    &\qquad \ge \Bigg( \sum_{f=0}^{\lfloor \tau^2 x - F \rfloor}  \binom{F}{f} \hat{p}^f (1-\hat{p})^{F-f}\Bigg)^{UM} \nn\\
    &\qquad\ge (1-\hat{p})^{FMU}, \nn
\end{align}
$\Pr[\zeta^{-2}\le x] = 1$ for $x\ge 2F\tau^{-2}$ and $\Pr[\zeta^{-2}\le x] = 0$ for $x< F\tau^{-2}$. Going back to (D.2), we can find the upper bound of it as 
\begin{align}
    \text{(D.2)} &\lr\le{lemlr05} \bbE[{\zeta}^{-2}]\nn \\
    &= \int_{0}^{\infty} 1 - \Pr[{\zeta}^{-2}\le x] dx \nn \\ 
    &= \int_{0}^{2F\tau^{-2}} 1 - \Pr[{\zeta}^{-2}\le x] dx \nn \\
    &\le (2 - (1-\hat{p})^{FMU}) \times F\tau^{-2}, \label{result_L_2}
\end{align}
where \eqref{lemlr05} holds due to $\zeta\le\bar{\zeta}$.
By substituting \eqref{result_L_1} and \eqref{result_L_2} to \eqref{result_L1L2}, we can obtain the desired result as follows.
\begin{align}
    &\bbE[\bar{\gamma}^{-2}] \le (2 - (1-\hat{p})^{FMU}) \nn \\
    &\qquad\times \frac{8\lambda^2 M E^2 F}{P\tau^2}
    (   \sigma^2 +  M  \delta^2 + \bbE[\lVert \nabla L(\bsth) \rVert^2]). \nn
\end{align}

\section{Proof of Theorem \ref{thm:thm01}}
By substituting the result of \textbf{Lemma \ref{lemma02}} to \textbf{Lemma \ref{lemma01}}, we can obtain
\begin{align}
&\bbE\bigg[\frac{1}{T}\sum_{t=1}^{T}\lVert \nabla L(\bsth^{t}) \rVert^2 \bigg] 
    \le \frac{8(L(\bsth^1) - L(\bsth^*))}{\lambda e^{-\tau} E T} \nn \\
    &\quad+32  \lambda^2  \kappa^2 M E^2 \sigma^2 + 32 \lambda^2  \kappa^2 M^2 E^2 \delta^2 \nn \\
    &\quad+\frac{32\lambda \kappa M E \sigma^2}{e^{-\tau}} + \frac{32\lambda \kappa M^2 E \delta^2}{e^{-\tau}} \nn \\ 
    &\quad + \frac{32\lambda\kappa N_0 M E F^2 N \omega\sigma^2}{P \tau^2 e^{-\tau}  } + \frac{32\lambda\kappa N_0 M^2 E F^2 N \omega\delta^2}{P \tau^2 e^{-\tau}  }\nn\\
    &\quad+\frac{32\lambda\kappa N_0 M E F^2 N \omega}{P \tau^2 e^{-\tau}  }\bbE\bigg[\frac{1}{T}\sum_{t=1}^{T}\lVert \nabla L(\bsth^{t}) \rVert^2 \bigg] \nn.
\end{align}
By rearranging the above result, we can obtain
\begin{align}
    &\bigg(1-\frac{32\lambda\kappa N_0 M E F^2 N \omega}{P \tau^2 e^{-\tau}  }\bigg)\bbE\bigg[\frac{1}{T}\sum_{t=1}^{T}\lVert \nabla L(\bsth^{t}) \rVert^2 \bigg] \nn \\
    &\le \frac{8(L(\bsth^1) - L(\bsth^*))}{\lambda e^{-\tau} E T} \nn \\
    &\quad+32  \lambda^2  \kappa^2 M E^2 \sigma^2 + 32 \lambda^2  \kappa^2 M^2 E^2 \delta^2 \nn \\
    &\quad+\frac{32\lambda \kappa M E \sigma^2}{e^{-\tau}} + \frac{32\lambda \kappa M^2 E \delta^2}{e^{-\tau}} \nn \\ 
    &\quad + \frac{32\lambda\kappa N_0 M E F^2 N \omega\sigma^2}{P \tau^2 e^{-\tau}  } + \frac{32\lambda\kappa N_0 M^2 E F^2 N \omega\delta^2}{P \tau^2 e^{-\tau}  }\nn,
\end{align}
and since $\lambda<\frac{P\tau^2 e^{-\tau}}{64\kappa N_0 M E F^2 N \omega}$, we can reach the desired result.




\ifCLASSOPTIONcaptionsoff
  \newpage
\fi

\bibliographystyle{IEEEtran}
\bibliography{IEEEabrv,FAS+OFL}

\begin{thebibliography}{10}
\providecommand{\url}[1]{#1}
\csname url@samestyle\endcsname
\providecommand{\newblock}{\relax}
\providecommand{\bibinfo}[2]{#2}
\providecommand{\BIBentrySTDinterwordspacing}{\spaceskip=0pt\relax}
\providecommand{\BIBentryALTinterwordstretchfactor}{4}
\providecommand{\BIBentryALTinterwordspacing}{\spaceskip=\fontdimen2\font plus
\BIBentryALTinterwordstretchfactor\fontdimen3\font minus \fontdimen4\font\relax}
\providecommand{\BIBforeignlanguage}[2]{{%
\expandafter\ifx\csname l@#1\endcsname\relax
\typeout{** WARNING: IEEEtran.bst: No hyphenation pattern has been}%
\typeout{** loaded for the language `#1'. Using the pattern for}%
\typeout{** the default language instead.}%
\else
\language=\csname l@#1\endcsname
\fi
#2}}
\providecommand{\BIBdecl}{\relax}
\BIBdecl

\bibitem{arXiv15_Konecny}
\BIBentryALTinterwordspacing
J.~Konečný, B.~McMahan, and D.~Ramage, ``Federated optimization: Distributed optimization beyond the datacenter,'' \emph{arXiv}, vol. abs/1511.03575, 2015. [Online]. Available: \url{https://arxiv.org/abs/1511.03575}
\BIBentrySTDinterwordspacing

\bibitem{SIGSAC15_Shokri}
R.~Shokri and V.~Shmatikov, ``Privacy-preserving deep learning,'' in \emph{Proceedings of the ACM SIGSAC Conference on Computer and Communications Security}.\hskip 1em plus 0.5em minus 0.4em\relax Denver, Colorado, USA: Association for Computing Machinery, 2015, pp. 1310--1321.

\bibitem{AISTATS17_Mcmahan}
B.~McMahan, E.~Moore, D.~Ramage, S.~Hampson, and B.~A. y.~Arcas, ``Communication-efficient learning of deep networks from decentralized data,'' in \emph{Proceedings of the International Conference on Artificial Intelligence and Statistics (AISTATS)}, Ft. Lauderdale, Florida, USA, 2017.

\bibitem{WFL01}
W.~Y.~B. Lim, N.~C. Luong, and D.~T.~H. et~al., ``Federated learning in mobile edge networks: a comprehensive survey,'' \emph{IEEE Communications Surveys \& Tutorials}, vol.~22, no.~3, pp. 2031--2063, 2020.

\bibitem{WFL02}
L.~U. Khan, W.~Saad, Z.~Han, E.~Hossain, and C.~S. Hong, ``Federated learning for internet of things: recent advances, taxonomy, and open challenges,'' \emph{IEEE Communications Surveys \& Tutorials}, vol.~23, no.~3, pp. 1759--1799, 2021.

\bibitem{WFL03}
A.~Imteaj, U.~Thakker, S.~Wang, J.~Li, and M.~H. Amini, ``A survey on federated learning for resource constrained iot devices,'' \emph{IEEE Internet of Things Journal}, vol.~9, no.~1, pp. 1--24, 2022.

\bibitem{WFL04}
S.~Niknam, H.~S. Dhillon, and J.~H. Reed, ``Federated learning for wireless communications: motivation, opportunities, and challenges,'' \emph{IEEE Communications Magazine}, vol.~58, no.~6, pp. 46--51, 2020.

\bibitem{WFL05}
S.~Abdulrahman, H.~Tout, H.~Ould-Slimane, A.~Mourad, C.~Talhi, and M.~Guizani, ``A survey on federated learning: the journey from centralized to distributed on-site learning and beyond,'' \emph{IEEE Internet of Things Journal}, vol.~8, no.~7, pp. 5476--5497, 2021.

\bibitem{CommFL01}
K.~Yue, R.~Jin, C.~W. Wong, and H.~Dai, ``Communication-efficient federated learning via predictive coding,'' \emph{IEEE Journal of Selected Topics in Signal Processing}, vol.~16, no.~3, pp. 369--380, 2022.

\bibitem{CommFL02}
Y.~Yang, Z.~Zhang, and Q.~Yang, ``Communication-efficient federated learning with binary neural networks,'' \emph{IEEE Journal on Selected Areas in Communications}, vol.~39, no.~12, pp. 3836--3850, 2021.

\bibitem{Mine01}
S.~Park and W.~Choi, ``Regulated subspace projection based local model update compression for communication-efficient federated learning,'' \emph{IEEE Journal on Selected Areas in Communications}, vol.~41, no.~4, pp. 964--976, 2023.

\bibitem{CommFL03}
M.~Elmahallawy, T.~Luo, and K.~Ramadan, ``Communication-efficient federated learning for leo constellations integrated with haps using hybrid noma-ofdm,'' \emph{IEEE Journal on Selected Areas in Communications}, vol.~42, no.~5, pp. 1097--1114, 2024.

\bibitem{AFL01}
K.~Yang, T.~Jiang, Y.~Shi, and Z.~Ding, ``Federated learning via over-the-air computation,'' \emph{IEEE Transactions on Wireless Communications}, vol.~19, no.~3, pp. 2022--2035, 2020.

\bibitem{AFL02}
M.~M. Amiri and D.~Gündüz, ``Machine learning at the wireless edge: distributed stochastic gradient descent over-the-air,'' \emph{IEEE Transactions on Signal Processing}, vol.~68, pp. 2155--2169, 2020.

\bibitem{AFL03}
------, ``Federated learning over wireless fading channels,'' \emph{IEEE Transactions on Wireless Communications}, vol.~19, no.~5, pp. 3546--3557, 2020.

\bibitem{AFL04}
N.~Zhang and M.~Tao, ``Gradient statistics aware power control for over-the-air federated learning,'' \emph{IEEE Transactions on Wireless Communications}, vol.~20, no.~8, pp. 5115--5128, 2021.

\bibitem{AFL05}
T.~Sery, N.~Shlezinger, K.~Cohen, and Y.~C. Eldar, ``Over-the-air federated learning from heterogeneous data,'' \emph{IEEE Transactions on Signal Processing}, vol.~69, pp. 3796--3811, 2021.

\bibitem{Mine02}
S.~Park and W.~Choi, ``On the differential privacy in federated learning based on over-the-air computation,'' \emph{IEEE Transactions on Wireless Communications}, vol.~23, no.~5, pp. 4269--4283, 2024.

\bibitem{Mine03}
J.-P. Hong, S.~Park, and W.~Choi, ``Base station dataset-assisted broadband over-the-air aggregation for communication-efficient federated learning,'' \emph{IEEE Transactions on Wireless Communications}, vol.~22, no.~11, pp. 7259--7272, 2023.

\bibitem{Mine04}
S.~Park and C.~W, ``Byzantine fault tolerant distributed stochastic gradient descent based on over-the-air computation,'' \emph{IEEE Transactions on Communications}, vol.~70, no.~5, pp. 3204--3219, 2022.

\bibitem{AirCompFL01}
L.~Qiao, Z.~Gao, M.~B. Mashhadi, and D.~Gündüz, ``Massive digital over-the-air computation for communication-efficient federated edge learning,'' \emph{IEEE Journal on Selected Areas in Communications}, vol.~42, no.~11, pp. 3078--3094, 2024.

\bibitem{AirCompFL02}
B.~Jiang, J.~Du, C.~Jiang, Z.~Han, A.~Alhammadi, and M.~Debbah, ``Over-the-air federated learning in digital twins empowered uav swarms,'' \emph{IEEE Transactions on Wireless Communications}, vol.~23, no.~11, pp. 17\,619--17\,634, 2024.

\bibitem{AirCompFL03}
F.~M.~A. Khan, H.~Abou-Zeid, and S.~A. Hassan, ``Deep compression for efficient and accelerated over-the-air federated learning,'' \emph{IEEE Internet of Things Journal}, vol.~11, no.~15, pp. 25\,802--25\,817, 2024.

\bibitem{FAS}
K.~K. Wong, A.~Shojaeifard, K.~F. Tong, and Y.~Zhang, ``Fluid antenna systems,'' \emph{IEEE Transactions on Wireless Communications}, vol.~20, no.~3, pp. 1950--1962, 2021.

\bibitem{FAS_p2p01}
L.~Tlebaldiyeva, G.~Nauryzbayev, S.~Arzykulov, A.~Eltawil, and T.~Tsiftsis, ``Enhancing qos through fluid antenna systems over correlated nakagami-m fading channels,'' in \emph{Proceedings of IEEE Wireless Communications and Networking Conference (WCNC)}, Austin, TX, USA, 2022.

\bibitem{FAS_p2p02}
M.~Khammassi, A.~Kammoun, and M.~S. Alouini, ``A new analytical approximation of the fluid antenna system channel,'' \emph{IEEE Transactions on Wireless Communications}, vol.~22, no.~12, pp. 8843--8858, 2023.

\bibitem{FAS_p2p03}
C.~Psomas, G.~M. Kraidy, K.~K. Wong, and I.~Krikidis, ``On the diversity and coded modulation design of fluid antenna systems,'' \emph{IEEE Transactions on Wireless Communications}, vol.~23, no.~3, pp. 2082--2096, 2024.

\bibitem{FAS_p2p04}
P.~Mukherjee, C.~Psomas, and I.~Krikidis, ``On the level crossing rate of fluid antenna systems,'' in \emph{Proceedings of IEEE International Workshop on Signal Processing Advances in Wireless Communications (SPAWC)}, Oulu, Finland, 2022.

\bibitem{FAS_MAC01}
K.~K. Wong and K.~F. Tong, ``Fluid antenna multiple access,'' \emph{IEEE Transactions on Wireless Communications}, vol.~21, no.~7, pp. 4801--4815, 2022.

\bibitem{FAS_MAC02}
K.~K. Wong, K.~F. Tong, Y.~Chen, and Y.~Zhang, ``Closed-form expressions for spatial correlation parameters for performance analysis of fluid antenna systems,'' \emph{Electronics Letters}, vol.~58, no.~11, pp. 454--457, 2022.

\bibitem{FAS_MAC03}
K.~K. Wong, D.~Morales-Jimenez, and K.~F. Tong, ``Slow fluid antenna multiple access,'' \emph{IEEE Transactions on Communications}, vol.~71, no.~5, pp. 2831--2846, 2023.

\bibitem{AirCompFAS01}
D.~Zhang, S.~Ye, M.~Xiao, K.~Wang, M.~D. Renzo, and M.~Skoglund, ``Fluid antenna array enhanced over-the-air computation,'' \emph{IEEE Wireless Communications Letters}, vol.~13, no.~6, pp. 1541--1545, 2024.

\bibitem{AirCompFAS02}
N.~Li, P.~Wu, B.~Ning, L.~Zhu, and W.~Mei, ``Over-the-air computation via 2-d movable antenna array,'' \emph{IEEE Wireless Communications Letters}, vol.~14, no.~1, pp. 33--37, 2025.

\bibitem{Assumption}
A.~Fallah, A.~Mokhtari, and A.~Ozdaglar, ``Personalized federated learning with theoretical guarantees: A model-agnostic meta-learning approach,'' in \emph{Advances in Neural Information Processing Systems (NeurIPS)}, 2020, pp. 3557--3568.

\bibitem{MNIST}
L.~Deng, ``The mnist database of handwritten digit images for machine learning research,'' \emph{IEEE Signal Processing Magazine}, vol.~29, no.~6, pp. 141--142, 2012.

\bibitem{CIFAR10}
A.~Krizhevsky, ``Learning multiple layers of features from tiny images,'' University of Toronto, Tech. Rep., 2009.

\end{thebibliography}

\end{document}